\documentclass[ecta,nameyear,draft]{econsocart}
\RequirePackage[colorlinks,citecolor=blue,linkcolor=blue,urlcolor=blue,pagebackref]{hyperref}

\usepackage{amsfonts, dsfont, amsmath, amsthm, amssymb}
\usepackage[mathscr]{euscript}
\usepackage{tikz,pgfplots}
\usepackage{booktabs,longtable,array}
\usepackage{comment, xcolor, tcolorbox}
\usepackage{subcaption}

\makeatletter

\@ifundefined{tableofcontents}{%
  \newcommand\tableofcontents{%
    \section*{Contents}%
    \begingroup
      \parindent\z@ \parskip\z@

      \providecommand*{\@dotsep}{4.5}%
      \def\l@section##1##2{\@dottedtocline{1}{0em}{1.5em}{##1}{##2}}%
      \def\l@subsection##1##2{\@dottedtocline{2}{1.5em}{2.3em}{##1}{##2}}%
      \def\l@subsubsection##1##2{\@dottedtocline{3}{3.8em}{3.2em}{##1}{##2}}%

      \def\contentsline##1##2##3##4{%
        \csname l@##1\endcsname{##2}{##3}%
      }%

      \@starttoc{toc}%
    \endgroup
  }%
}{%
}

\makeatother

\startlocaldefs

\numberwithin{equation}{section}
\theoremstyle{plain}

\newtheorem{assumption}{Assumption}
\newtheorem{theorem}{Theorem}

\newtheorem{cor}{Corollary}

\theoremstyle{definition}
\newtheorem{definition}{Definition}
\newtheorem{runexample}{Running Example}
\newtheorem{example}{Example}
\newtheorem{remark}{Remark}
\newtheorem{auxlemma}{Lemma}[section]

\let\originallesssim\lesssim
\let\originalgtrsim\gtrsim

\DeclareRobustCommand{\lesssim}{%
  \mathrel{\mathpalette\lowersim\originallesssim}%
}
\DeclareRobustCommand{\gtrsim}{%
  \mathrel{\mathpalette\lowersim\originalgtrsim}%
}

\makeatletter
\newcommand{\lowersim}[2]{%
  \sbox\z@{$#1<$}%
  \raisebox{-\dimexpr\height-\ht\z@}{$\m@th#1#2$}%
}
\makeatother

\DeclareMathOperator*{\argmax}{arg\,max}

\newcommand{\independent}{\rotatebox[origin=c]{90}{$\models$}}

\newcommand{\ra}{\rightarrow}
\newcommand{\imp}{\Rightarrow}
\newcommand{\R}{\mathbb{R}}
\newcommand{\one}{\mathds{1}}

\newcommand{\bT}[1]{{{\mathbf T}^{(#1)}}}
\newcommand{\bY}[1]{{\mathbf Y}^{(#1)}}

\newcommand{\bt}[1]{{\mathbf t}^{(#1)}}
\newcommand{\tr}[1]{{t}^{(#1)}}
\newcommand{\Y}[1]{{Y}^{(#1)}}
\newcommand{\Xc}{\mathcal X}
\newcommand{\Gc}[1]{{\mathcal G}_{(#1)}}
\newcommand{\Yc}[1]{{\mathcal Y}^{(#1)}}
\newcommand{\Tc}{{\mathcal T}}
\newcommand{\bX}[1]{{\mathbf{X}^{(#1)}}}

\newcommand{\ex}{\mathbb E}
\newcommand{\p}{\mathbb P}

\newcommand{\Pc}{\mathcal{P}}
\newcommand{\bA}[1]{{\mathbf A}^{(#1)}}
\newcommand{\bG}[1]{{\mathbf G}^{(#1)}}

\newcommand{\sA}{\mathscr A}

\newcommand{\dx}{{d_x}}

\newcommand{\J}{{n_0}}
\newcommand{\B}{B}
\newcommand{\X}[1]{X^{(#1)}}
\newcommand{\G}[1]{G^{(#1)}}
\newcommand{\W}{W}

\newcommand{\Q}{{Q_{0}}}
\newcommand{\pig}{\pi}
\newcommand{\pil}{\lambda}
\newcommand{\pigpil}{{\pig}_{\pil}}

\newcommand{\supp}[1]{\text{supp}_\p(#1)}

\newcommand{\cC}{\mathcal C}
\newcommand{\eft}{{c}^{(e)}}

\newcommand{\sigalg}{\mathcal{F}}
\newcommand{\n}{{n_1}}

\newcommand{\nee}{{n_e}}

\endlocaldefs

\begin{document}

\begin{frontmatter}

\title{Optimal Treatment Assignment Rules on Networked Populations}
\runtitle{Treatment Rules on Networks}

\begin{aug}
\author[add1]{\fnms{Abhishek}~\snm{Ananth}\ead[label=e1]{abhishek.ananth@emory.edu}}

\address[add1]{%
\orgdiv{Department of Data and Decision Sciences},
\orgname{Emory University}}
\end{aug}

\begin{funding}
I thank Francesca Molinari, J\"org Stoye, Larry Blume, David Easley and Aleksey Tetenov for helpful comments and discussions. I also thank seminar participants at University of Toronto, University of Fribourg, University of Bern, CREST-PSE, UCL, University of Amsterdam, University of Oxford, University of Bonn, NYU, Northwestern University,  University of Chicago and Harvard-MIT Econometrics Seminar as well as conference participants at WEAI 2019 for useful feedback. 
\end{funding}


\begin{abstract}
I analyze decision problems where a planner must choose, based on data from an experiment in a pilot population, how to simultaneously distribute a limited number of treatments that generate network spillovers across individuals in a target population.
A \emph{global treatment rule} maps the target population's observed covariates and network features into a jointly feasible allocation. 
Under local-response, stability, support, and policy-relevant overlap conditions, I identify the resulting welfare using a transported inverse propensity weighted criterion and construct an empirical welfare maximizer.
I derive a finite-sample regret bound under restrictions on the complexity of network spillovers and the complexity of the global rule class, measured by its Natarajan dimension.
The upper bound uses independent node-level variation, dyadic network formation without imposing a bounded maximum degree.
A minimax lower bound establishes a corresponding rate comparison under additional welfare-relevant richness. 
Finally, I compare oracle optimal global and local rules (which allocate
treatments on an individual-by-individual basis) and bound the welfare gap between them. An empirically calibrated oracle comparison
illustrates the magnitude of this welfare gap.
\end{abstract}

\end{frontmatter}

\section{Introduction}
\label{sec:intro}

Policymakers are frequently tasked with allocating a limited supply of treatments across individuals in a population whose outcomes depend not only on their own treatment but also on that of others. A salient example arises in the immediate aftermath of outbreaks of infectious diseases, when treatments are scarce and must be rationed under severe logistical and ethical constraints.\footnote{Treatments encompass both preventive and therapeutic interventions, such as pre-exposure prophylaxis (PrEP) for HIV, antiviral drugs like oseltamivir for influenza, or mRNA vaccines for COVID-19.}
During the West and Central African Ebola outbreaks, for instance, the deployment of the rVSV-ZEBOV vaccine represented a major breakthrough \citep{henao2015efficacy, henao2017efficacy}.
The World Health Organization adopted an experimental ``ring vaccination'' strategy that had previously proved successful in eradicating smallpox \citep{fenner1988smallpox}.\footnote{Ring vaccination inoculates individuals who had direct contact with Ebola patients as well as their contacts.}
Although this strategy helped contain the outbreak, whether it achieved an optimal allocation of scarce vaccines remains unclear \citep{merler2016containing, kucharski2016effectiveness, hitchings2017using}.
Humanitarian organizations, including Doctors Without Borders, criticized the slow distribution, unclear prioritization criteria, and, consequently, the overall effectiveness of the strategy \citep{MSF:Key}.
Such examples underscore the urgent need for data-driven treatment allocation strategies tailored to complex environments with interdependent outcomes and capacity constraints.\footnote{Similar allocation challenges arise in viral marketing, where decision-makers select a subset of individuals to receive a product first, leveraging their social influence to encourage adoption among peers \citep{kempe2003maximizing}.}

Motivated by these examples, I develop a framework in which a planner constructs a \emph{global treatment rule}, which maps an observed \emph{target} population into a complete allocation vector, when key features of the decision environment—including heterogeneity in treatment effects and spillovers—are unknown and must be learned using experimental data from a distinct \emph{pilot} population.\footnote{Heterogeneous treatment effects occur when individuals respond differently to the same treatment. Spillover effects arise when an individual’s outcome depends on treatments received by others.}
In so doing, I extend \emph{empirical welfare maximization} (EWM) to the global treatment allocation problem, where the \emph{welfare} is a population-level metric, such as the expected number of survivors following treatment allocation, that the planner wishes to optimize by appropriately choosing a global treatment rule.

Following the literature on treatment allocation under heterogeneous effects \citep{manski2004statistical, hirano2009asymptotics, bhattacharya2012inferring, kittet18}, I allow treatment decisions to depend on pre-treatment covariates observed by the planner. However, this literature typically imposes Stable Unit Treatment Value Assumption (SUTVA) \citep{rubin1980randomization}, which rules out spillovers. To capture spillovers, I model interactions through networks, where individuals are nodes and edges represent contacts that facilitate spillovers, such as disease transmission. Treatment rules may therefore depend not only on observed covariates but also on observed network features.

The proposed EWM procedure evaluates candidate global rules using an inverse propensity weighted (IPW) estimator of welfare. The estimator transports experimental evidence from the pilot population to the target population for a broad class of low-dimensional exposure mappings. This class includes local and exchangeable spillovers (also called anonymous interactions) which are common in the social sciences, while Appendix~\ref{app:sir} illustrates the network exposure resulting from a finite-horizon SIR model, a workhorse model in epidemiology.
I evaluate the performance of the empirical welfare maximizer using worst-case regret, a standard criterion in statistical learning theory \citep{devroye1996uniform, kittet18}. Regret compares the expected welfare, averaged over repeated samples, under the rule selected by EWM to the welfare under the optimal global treatment rule. The identification and estimation of the welfare criterion is developed in Section~\ref{sec:setup}.

Section~\ref{sec:regret-bounds} establishes finite-sample upper and lower bounds on the regret resulting from EWM. Global allocation creates a multiclass learning problem: the output is a treatment vector, and the number of possible outputs can grow rapidly with target-population size. I measure the complexity of the candidate rule class separately through its Natarajan dimension \citep{natarajan1989}.
Network interference creates a further challenge because the observed welfare contributions are dependent. My analysis uses independent node-level fundamentals and dyadic randomization, together with locality restrictions controlling how complete welfare contributions change when node fundamentals or network links are perturbed. The graphon sparsity parameter bounds expected degree rather than imposing a hard maximum-degree restriction, which adds to a literature that models dependence through dependency graphs \citep[see][]{janson2004, zhang2024generalization, viviano2025policy}.

Under the above restrictions, I derive an upper bound on regret that decreases with the size of the pilot population, which reflects the value of additional experimental information. However, the bound also increases with the size of the target population, through both the number of feasible treatment exposures as well as the Natarajan dimension of the class of global rules. Thus, whether regret shrinks to zero (and hence whether learning of the welfare maximizing rule occurs) depends explicitly on the relative growth rates of the pilot and target populations. I also derive a minimax lower bound on regret. The upper and lower bounds match in their dependence on the pilot size, up to logarithmic factors, but differ in how they scale with the target size. This rate comparison requires an additional richness of the class of candidate global rules to handle the global capacity constraint. 

This paper's focus on global treatment rules contrasts with the existing literature on optimal allocations under spillovers, which restrict attention to \emph{local treatment rules} \citep{zhang2023individualized, kitagawa2023should, viviano2025policy, dean2025effective}. Local rules assign treatment to each individual based on its own local information (e.g., observed covariates and network features) and thereby impose a \emph{symmetry constraint}: any two individuals with identical local information must be assigned the same treatment. This restriction can substantially simplify the learning problem and may be natural when policies are implemented through individual-level eligibility rules.

However, when the marginal welfare benefit from treating individuals is positive in some observed states and negative in others, the optimal allocation often requires coordinating treatment decisions across individuals, for example treating one individual while withholding treatment from another with identical local information in order to induce favorable conditions overall. Global rules can implement such joint, asymmetric allocations, whereas local rules cannot. Consequently, the optimal local rule can be suboptimal relative to the optimal global rule, resulting in a welfare loss. I also characterize how this welfare loss scales with key features of the decision environment such as the heterogeneity in the sign of direct benefits and spillovers to
other units (or, more precisely, budget-adjusted social returns) net of the opportunity cost of satisfying the capacity constraint. Using empirically relevant examples, I show that the welfare loss can be substantial in practice. The results therefore characterize both the potential welfare gains from moving beyond local rules and the statistical cost of learning favorable global treatment rules from finite experimental data in the presence of spillovers.

\subsection{Related Literature}
\label{sec:lit_rev}
\citet{manski2004statistical} places treatment assignment under uncertainty within a statistical decision-theoretic framework, emphasizing minimax rules when potential outcomes are only partially known. This perspective is closely related to empirical risk minimization in statistical learning, where decision rules are chosen to minimize expected loss over a restricted class. Learning-theoretic tools such as VC dimension, covering numbers, and chaining yield finite-sample excess-risk bounds for classification \citep{devroye1996uniform,boucheron2005theory,cesa2006prediction,bartlett2002rademacher,tsybakov2004optimal,audibert2007fast}. \citet{kittet18} adapt these ideas to treatment allocation through empirical welfare maximization, while \citet{athey2021policy} consider continuous treatments and \citet{viviano2025policy} study local policy rules under network interference. My analysis combines IPW estimation with empirical-process and graphon network formation arguments to accommodate general exposure mappings and transport welfare from a pilot population to a target population \citep{dudik2011policy,agarwal2014taming,chernozhukov2018double,janson2004,boucheron2013concentration,vdvwellner1996}.
The paper is also related to work on structured and network-dependent allocation. \citet{hazard2025whom} learn jointly constrained matching policies, and \citet{nie2021learning} study dynamic rules governing whom and when to treat. Closer to my setting, \citet{wang2024robust} and \citet{kitagawa2023individualized} analyze targeting in simultaneous and sequential network games. These papers derive spillovers from a specified strategic model, whereas I use a general exposure-mapping framework and study the additional problem of transferring a global allocation rule across populations.

A complementary literature modifies the planner's objective to account for forms of uncertainty not captured by conventional welfare maximization. \citet{chernozhukov2025policy} develop confidence-aware policy rules, \citet{ishihara2021evidence} study minimax-regret aggregation of evidence across populations, \citet{kitagawa2023treatment} consider mean-square regret under partial identification, and \citet{kitagawa2025leave} incorporate aversion to inequality in individual regret. In contrast, I retain an expected-welfare criterion and focus on the identification and statistical costs generated jointly by interference, global treatment constraints, and differences between pilot and target environments.

\section{Setup}
\label{sec:setup}

For the reader's convenience, Appendix~\ref{app:notation} contains a table summarizing the notation used throughout the paper.
All random objects are defined on a common probability space
$(\Omega,\sigalg_\Omega,\p)$. All covariate and observed-feature spaces
are standard Borel, with their Borel sigma-algebras. 
Finally, for any positive integer $k$, let $[k] := \{1,\ldots,k\}$.

Suppose there exist two populations: a \emph{target} population (denoted by $e = 0$) and a \emph{pilot} population (denoted by $e = 1$). 
Each population $e \in \{0,1\}$ consists of $\nee$ units which may represent households, individuals, or geographic sites and are indexed by $j \in \{1, \cdots, \nee\}$.
A planner observes a $d_x$-dimensional covariate vector $\X{e}_{j}\in\Xc\subset\R^{\dx}$ for each unit. 
The matrix of observed population-level covariates is denoted by $\bX{e} = [\X{e}_{j}]_{j=1}^{\nee}$.
Each unit can be assigned either to the treatment ($\tr{e}_j = 1$) or to the control ($\tr{e}_j = 0$) condition.\footnote{This framework can be extended to multi-valued treatment doses with more cumbersome notation.}  
Denote by $\bt{e}=[\tr{e}_j]_{j=1}^{\nee}\in\{0,1\}^{\nee}$ the treatment allocation vector in population $e$. Unit $j$'s treatment response or outcome (denoted by $\Y{e}_j$) may violate the commonly imposed Stable Unit Treatment Value Assumption (SUTVA) and depends on the treatments assigned to others in the population (spillover) in addition to their own treatment status (direct effect). Further, a units treatment response may also be heterogeneous in its observed covariate.

\begin{runexample} 
A population (village) consisting of $\nee$ households faces an outbreak of malaria, a life-threatening disease typically transmitted to humans by female \emph{Anopheles} mosquitoes. These mosquitoes breed in vegetation-rich areas with stagnant water and typically travel only a few hundred meters \citep{midega2007estimating, cowman2016malaria}. Each household’s observed covariate is its distance to the nearest patch of dense vegetation or pool of stagnant water, $\X{e}_j \in [0,500]$.
Female mosquitoes are nocturnal and feed on human and livestock blood to produce and lay eggs; consequently, insecticide-treated bed-nets are a cost-effective preventive intervention \citep{curtis2006insecticide}. Let $\tr{e}_j = 1$ indicate that household $j$ is assigned a bed-net, and $\tr{e}_j = 0$ indicate that it does not receive a bed-net and instead relies on status quo mitigation practices. Households located closer to vegetation (i.e., with smaller $\X{e}_j$) face higher infection risk and therefore derive greater expected benefit from bed-net use, leading to heterogeneity in treatment response. 
During their activity, mosquitoes sometimes transmit malarial parasites from one infected individual to another. 
Therefore, when households are close enough for mosquitoes to travel between them, the treatments induce spillovers (also known as interference).
$\blacktriangledown$ \end{runexample}

\subsection{Spillovers}

I assume spillovers are mediated through an undirected, unweighted network represented by an adjacency matrix $\bA{e}$, an $\nee \times \nee$ symmetric matrix of binary valued entries with no self links. Denote by $\sA_{\nee}$ the set of all such adjacency matrices. 
I assume the network is exogenous and static, not responding to the allocation of treatments. 
This fixed-network assumption is common in the literature on network interference \citep{aronow2017estimating, leung2020treatment, forastiere2021identification} but can be restrictive in some contexts.\footnote{\citet{comola2021treatment, kitagawa2023should, gao2024endogenous} are notable exceptions that allow the network to respond to treatment allocation.} For example, in HIV/AIDS settings, contact behavior (interactions that facilitate transmission) may respond to the availability and use of pre-exposure prophylaxis (PrEP), a preventative treatment against HIV \citep{walsh2020individual}. Assumptions on the network formation model follow the definition of treatment exposures and the planner's decision problem.

Following much of the literature, I assume the existence of an \emph{effective treatment}, a network-dependent low-dimensional summary of the treatment assignment vector that determines each units outcome \citep{manski2013identification}. I formalize this notion through an \emph{exposure mapping} as defined in \citet{aronow2017estimating}. I start by defining the space of exposures.

\begin{definition}[Space of exposures]
    Let $\cC^{(e)}$ denote the finite space of possible exposures (or effective treatments) for units in population $e$. 
\end{definition}

I assume that $j$'s effective treatment depends on the network through an observed local network feature $\G{e}_j \in \Gc{e} \subset \R^{d_g^{(e)}}$, a measurable transformation of the network and covariates $(\bA{e},\bX{e})$.
This may be restrictive in some cases. Much of the network interference literature \citep{manski2013identification, aronow2017estimating, leung2020treatment} typically assumes that the full network is observed and consequently so are the local network features. When the full network is observed, a natural way to record the information is by setting the observed network feature for $j$ to be the $j^{\text{th}}$ row of the adjacency matrix, i.e. $\G{e}_j := \bA{e}_{j\cdot} \in \R^{\nee}$. Stack the observed network features as $\bG{e} := (G_1^{(e)},\ldots,G_{\nee}^{(e)})$.
When the planner partially observes the network through $\G{e}_j$, Assumption~\ref{asmp2::eff_treat} guarantees its sufficiency for computing exposures. Potential outcomes in the presence of spillovers are indexed by exposures as defined in Assumption~\ref{asmp2::pot_out}.

\begin{assumption}[Exposure Mapping]\label{asmp2::eff_treat}
    For any unit $j$ in population $e\in\{0,1\}$, the exposure mapping is some known function $\eft_j: \{0,1\}^{\nee}\times \Gc{e} \ra \cC^{(e)}$. 
\end{assumption} 

I denote the observed population state of a population by the shorthand $V$:
\begin{equation*}\label{eq:state-V}
 V^{(e)} := (\bX{e},\bG{e})\in\mathcal V_e,
 \qquad \mathcal V_e := \Xc^{\nee}\times\Gc{e}^{\nee},
 \qquad V := V^{(0)},\quad\mathcal V := \mathcal V_0.
\end{equation*}
Thus $V$ is the observed target-population state, while $V^{(1)}$ is the observed pilot-population state. 
A lower-case $v$ denotes a realization, and $X_j^{(e)}(v),G_j^{(e)}(v)$ denote its components. 
I use $E$ as a shorthand for an exposure value, while $c_j^{(e)}$ denotes the known mapping itself:
\begin{equation*}\label{eq:exposure-E}
 E_j^{(e)}(\mathbf t;v)
  := \eft_j\bigl(\mathbf t,G_j^{(e)}(v)\bigr),
 \qquad
 E_j^{(e)}(\mathbf t) := E_j^{(e)}(\mathbf t;V^{(e)}).
\end{equation*}
The argument after the semicolon specifies a fixed population state; omitting
it means that the realized state of population $e$ is used.

\begin{assumption}[Potential Outcomes]\label{asmp2::pot_out}
For any unit $j$ in population $e\in\{0,1\}$:
\begin{enumerate}
    \item[a.] Potential outcomes are indexed by exposures and the family of all potential outcomes is denoted by $\Yc{e}_j ~:=~ \{\Y{e}_j(c): c \in \cC^{(e)} \}$.
    \item[b.] Potential outcomes are uniformly bounded. There exists $M < \infty$ such that $|\Y{e}_j(c)| \leq M $ almost surely for any exposure $c \in \cC^{(e)}$.
\end{enumerate}
\end{assumption}

Assumption~\ref{asmp2::pot_out}.a requires that any two allocations inducing the same exposure for $j$ generate the same treatment response for $j$. This framework covers standard exposure models \citep{aronow2017estimating, leung2020treatment, forastiere2021identification} and may also encode cluster-based partial interference with a sufficiently rich finite exposure space, as in \citep{hudgens2008toward, zhang2023individualized, dean2025effective}.
Assumption~\ref{asmp2::pot_out}.b requires that, regardless of the exposure, each individual's potential outcome lies within fixed bounds, ruling out arbitrarily large gains or losses.\footnote{This is for analytical convenience; similar results can be derived under finite-moment assumptions or sub-Gaussian tails with different analytical tools to obtain comparable bounds.}

I define the collections of feasible exposures for a node given its  network information:
\begin{definition}[Feasible exposure set]\label{def2::feas_exp}
    Let the correspondence $\cC^{(e)}_j: \Gc{e} \rightrightarrows \cC^{(e)}$ represent the collection of exposures that are feasible for an individual given its local network information. This is the exposure-relevant information in the observed network features.
    \begin{equation*}
        \cC^{(e)}_j(g) ~:=~ \left\{ \eft_j(\bt{e}, g) ~:~ \bt{e} \in \{0,1\}^{\nee} \right\}
    \end{equation*}
\end{definition}
For brevity, denote the realized feasible set and the exposure-relevant local
state by
\begin{equation*}\label{eq:local-state-L}
 C_j^{(e)} := \cC_j^{(e)}(\G{e}_j),
 \qquad L_j^{(e)} := (\X{e}_j,C_j^{(e)}).
\end{equation*}
Write $\ell=(x,C)$ for a generic local state. Both populations' local states
are embedded in
$\mathcal L := \Xc\times2^{\cC^{(0)}\cup\cC^{(1)}}$, equipped with the
product sigma-algebra. Since the exposure space is finite and the exposure
mapping is measurable, $C_j^{(e)}$ is a measurable, nonempty finite random
set. The shorthand $L_j^{(e)}$ records the network information used by outcomes and treatment propensity among other things. For a fixed state $v$, write
\[ C_j^{(e)}(v) := \cC_j^{(e)}(G_j^{(e)}(v)), \qquad L_j^{(e)}(v) := (X_j^{(e)}(v),C_j^{(e)}(v)). \]

\begin{runexample}\label{ex_2}
Due to spillovers resulting from vector-borne transmissions, treating one household may also reduce transmission to nearby households. This proximity of households, so mosquitoes can fly between them, is captured in a contact network, $\mathbf A$.   
The spillover mechanism described above can be summarized by a common specification \emph{local spillovers with exchangeability} (LSE), where household $j$'s exposure is summarized by own treatment, the number of treated neighbors, and degree:
\begin{equation}\label{eg::lse}
	c_j(\mathbf t, G_j) ~=~ \Bigl( t_j, \sum_k A_{jk} t_k, \sum_k A_{jk} \Bigr), \quad G_j = \{\mathbf A_{ij}: 1 \leq i \leq \nee\}.
\end{equation}
This allows spillovers to be non-monotone in treated neighbors (e.g., due to behavioral adaptation by mosquitoes), as it imposes no structure across exposure levels.
In this context, Assumption~\ref{asmp2::eff_treat} also requires that $G_j$ identifies $j$'s neighbor set (not only degree), as in snowball-style designs. Under LSE, the feasible exposure set $C_j = \{0,1\} \times \left\{0, \ldots, \sum_k A_{jk}\right\} \times \bigl\{ \sum_k A_{jk} \bigr\}$ is determined by node degree. 

Let $Y_j(\mathbf t)$ denote the (binary) potential health status of household $j$ one year after allocation $\mathbf t\in\{0,1\}^{\nee}$, where $1$ indicates healthy and $0$ unhealthy (due to infection or otherwise). The collection of potential outcomes for household $j$ satisfying Assumption~\ref{asmp2::pot_out} are:
\[ \mathcal Y_j= \big\{Y_j(t,d,l): ~ \text{ where } t\in \{0,1\}, ~ 0\leq d\leq l\leq n_e-1,~\text{ and }l,d\in\{0,1,\ldots,n_e-1\} \big\}. \qquad \blacktriangledown \]
\end{runexample}

In addition to the spillover model considered in the example above, the results in this paper extend to other important settings. In the appendix, I show how compartmental models, such as the benchmark SIR model used to analyze the transmission of infectious diseases can be formalized to satisfy Assumption~\ref{asmp2::eff_treat}. Different spillover models induce different patterns of interdependence in outcomes and, consequently, have distinct implications for learnability in the treatment assignment problem introduced in Section~\ref{sec:dec_prob}. It is therefore useful to introduce a measure of the complexity of a given spillover model which characterizes the sensitivity of exposures across all nodes to small perturbations in the network. 

\begin{assumption}[Locality of the complete welfare contribution]
\label{asmp:ub-locality}
\ 
\begin{enumerate}
\item[a.] There are two distinct exposure labels
$c_{\varnothing,0}\neq c_{\varnothing,1}$, common to both populations.
For either population, whenever node $k$ is isolated,
\[
 C_k^{(e)}=C_\varnothing:=\{c_{\varnothing,0},c_{\varnothing,1}\},
 \qquad E_k^{(e)}(\mathbf t)=c_{\varnothing,t_k}.
\]
For every $i\neq k$, the pair $(L_i^{(e)},E_i^{(e)}(\mathbf t))$
is unchanged when $(x_k,t_k)$ is altered, holding the graph and all other
covariates and treatments fixed.
\item[b.] There exists $d_\nee^{\mathrm{sc}}\in\{0,\ldots,\nee\}$ such that,
for all pilot graphs $A,A'$ differing in one dyad and every covariate
array $\mathbf x\in\Xc^\nee$,
\[
 \sup_{\mathbf t\in\{0,1\}^{\nee}}
 \left|\left\{i:\
 (L_i^{(e)}(A,\mathbf x),E_i^{(e)}(\mathbf t;A,\mathbf x))
 \neq (L_i^{(e)}(A',\mathbf x),E_i^{(e)}(\mathbf t;A',\mathbf x))
 \right\}\right|\leq d_\nee^{\mathrm{sc}}.
\]
\end{enumerate}
\end{assumption}
Assumption~\ref{asmp:ub-locality}.a states that isolated nodes have exactly two exposures, corresponding to the two states of self-treatment. 
It also rules out effects of an isolated node's covariates or treatment on
other units' exposure-relevant information. Part b controls the number of nodes whose feasible exposure sets as well as realized exposures change as a result of perturbing a single link in the network. The parameter $d_\nee^{\mathrm{sc}}$ can be interpreted as the dimension of spillover complexity for the assumed exposure mapping.

\begin{runexample}
Under LSE, toggling a link affects only its two endpoints' feasible sets and
exposures. Both parts of Assumption~\ref{asmp:ub-locality} hold with
$d_\nee^{\mathrm{sc}}=2$ when $\nee\geq2$. For own-treatment-only exposures,
$d_\nee^{\mathrm{sc}}=0$ is valid. These bounds remain valid when evaluating
the empty graph and all single-dyad perturbations. $\blacktriangledown$
\end{runexample}

\subsection{Treatment Allocation Problem}\label{sec:dec_prob}

A planner must allocate at most $\B$ units of the treatment among the $\J$ units in the target population. Denote the set of feasible allocations with cardinality $\Q = \sum_{b=0}^{\B} \binom{\J}{b}$ by:
\begin{equation}
    \Tc = \left\{\mathbf t \in \{0,1\}^{\J} : \sum_{j=1}^{\J} t_j \le \B \right\}.
\end{equation} 

The planner uses the observed covariate–network configuration in the target population while allocating treatments. 
This is formalized by a \emph{treatment allocation rules}, denote by $\pi$. 
The collection of all such {treatment allocation rules} is denoted by $\Pi$, formalized below: 
\begin{equation}
\Pi \ \subset\ \overline{\Pi} \ :=\  \left\{\, \pi : \mathcal V_{0} \rightarrow \Tc\right\}
\end{equation}

For a target rule $\pi$, I denote its induced exposure by the convenient shorthand:
\begin{equation}\label{eq:policy-exposure-E}
 E_j^\pi := E_j^{(0)}\bigl(\pi(V);V\bigr),
 \qquad E_j^\pi(v) := E_j^{(0)}\bigl(\pi(v);v\bigr).
\end{equation}
The planner evaluates allocation rules using a \emph{welfare} criterion $W(\cdot;\p): \overline \Pi \to \R$, which corresponds to the expected average outcome under the allocation prescribed by $\pi$:\footnote{The average is taken over units in the target population $[\J]$, and the expectation  with respect to DGP $\p$.}
\begin{align}\label{eq2:welfare}
  \W(\pi; \p) &= \ex \left[ \frac{1}{\J} \sum_{j=1}^{\J} \Y{0}_j(E_j^\pi)  \right] 
\end{align}
In the ideal problem, the planner chooses the \emph{welfare maximizing allocation rule}:
\begin{equation}\label{eq2:planner}
  \pi^*_{\p} ~\in~ \arg\sup_{\pi \in \Pi} \ \W(\pi; \p), \qquad\text{when the supremum is attained.}
\end{equation} 

I assume that $\Pi$ is nonempty and use
$\sup_{\pi\in\Pi}\W(\pi;\p)$ as the oracle welfare. Bounded outcomes make
this supremum finite, but a finite allocation set does not by itself imply
attainment of the supremum over an arbitrary policy class.

\begin{runexample}
Consider a target village with $\J=2$ households where the planner must allocate $\B=1$ bed-net. 
The contact network is completely determined by either off-diagonal element of the adjacency matrix $A_{12}$, itself characterized by the distance between the households. Treatment allocation rules map observed covariate–network configuration to the set of feasible allocations, $\Tc=\{(1,0),(0,1),(0,0)\}$:
\[
\pi:\ V=(\mathbf X,A_{12})\mapsto \mathbf t=(t_1,t_2)\in\Tc,
\qquad \mathbf X=(X_1,X_2)^\top.
\]
Welfare at allocation rule $\pi$ equals expected average health one year after assignment:
\[ \W(\pi;\p)
= \tfrac{1}{2} ~ \ex\!\left[ Y_1(E_1^\pi) +  Y_2(E_2^\pi)\right] \]
Suppose $\Pi$ is the collection of all functions from $\Xc^2 \times \{0,1\}$ mapping into $\Tc$. I illustrate the welfare maximizing rule for a simple DGP $\p$. Consider the following outcome model in which heterogeneity and spillovers play a role:
\[
\ex\!\left[ Y_j(E_j(\mathbf t; v)) \mid V = v \right]
= (1-\alpha)+\beta\, t_j d_j+\gamma\,A_{12}t_{3-j},
\qquad
d_j := \frac{250-X_j}{250},
\]
where $\alpha\in(0,1)$ is baseline infection risk, $\beta>0$ is the direct effect (larger when closer to vegetation, i.e.\ $d_j$ larger). Importantly, this specification allows that when a household is far from vegetation and at a lower risk from mosquitoes, its treatment response is negative allowing for the treatment to exhibit negative side effects such as risk of respiratory symptoms or skin irritation. Finally, $\gamma>0$ indexes spillovers when households are linked. 
\begin{figure}[ht]
\centering
\begin{tikzpicture}
\begin{axis}[
    title={Households not linked: $A_{12}=0$},
    xlabel={Distance of HH 1 from vegetation},
    ylabel={Distance of HH 2 from vegetation},
    xmin=0, xmax=500,
    ymin=0, ymax=500,
    scale=0.8,
    every axis plot/.append style={ultra thick}
]

\fill[blue!50!black]  (0,0)--(500,0)--(500,250)--(250,250)--cycle;
\fill[blue!50!white]  (0,0)--(0,500)--(250,500)--(250,250)--cycle;
\fill[black!62!white] (250,250)--(250,500)--(500,500)--(500,250)--cycle;   
\draw[dashed, white, thick] (axis cs:250,0) -- (axis cs:250,500);
\draw[dashed, white, thick] (axis cs:0,250) -- (axis cs:500,250);
\node[draw,align=left,scale=1,fill=white] at (300,100) {\textbf{\color{blue!50!black}Treat}\\\textbf{\color{blue!50!black}HH 2 only}};
\node[draw,align=left,scale=1,fill=white] at (100,300) {\textbf{\color{blue!50!white}Treat}\\\textbf{\color{blue!50!white}HH 1 only}};
\node[draw,align=left,scale=1,fill=white] at (400,395) {\textbf{\color{black!62!white}Treat}\\\textbf{\color{black!62!white}neither}};
\end{axis}
\end{tikzpicture}
\begin{tikzpicture}
\begin{axis}[
    title={Households linked: $A_{12}=1$},
    xlabel={Distance of HH 1 from vegetation},
    xmin=0, xmax=500,
    ymin=0, ymax=500,
    scale=0.8,
    every axis plot/.append style={ultra thick}
]

\fill[blue!50!black]  (0,0)--(500,0)--(500,350)--(350,350)--cycle;
\fill[blue!50!white]  (0,0)--(0,500)--(350,500)--(350,350)--cycle;
\fill[black!62!white] (350,350)--(350,500)--(500,500)--(500,350)--cycle; 
\draw[dashed, white, thick] (axis cs:350,0) -- (axis cs:350,500);
\draw[dashed, white, thick] (axis cs:0,350) -- (axis cs:500,350);

\node[draw,align=left,scale=1,fill=white] at (300,100) {\textbf{\color{blue!50!black}Treat}\\\textbf{\color{blue!50!black}HH 2 only}};
\node[draw,align=left,scale=1,fill=white] at (100,300) {\textbf{\color{blue!50!white}Treat}\\\textbf{\color{blue!50!white}HH 1 only}};
\node[draw,align=left,scale=1,fill=white] at (425,420) {\textbf{\color{black!62!white}Treat}\\\textbf{\color{black!62!white}neither}};
\end{axis}
\end{tikzpicture}
\caption{Optimal treatment rule under the known DGP.}
\label{fig::optimalrule}
\end{figure}
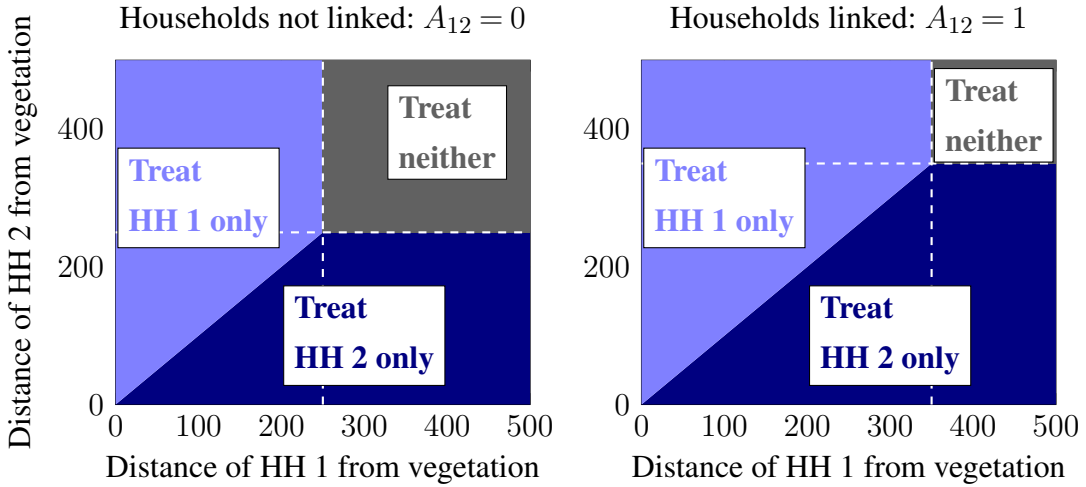

Figure~\ref{fig::optimalrule} plots the optimal allocation in $(X_1,X_2)\in[0,500]^2$ when the households are (right panel) and are not (left panel) linked. When $A_{12}=0$, it is optimal to treat neither household if both are farther than $250$ meters from vegetation; otherwise treat the closer household. When $A_{12}=1$, spillovers shrink the ``treat neither'' region: it is optimal to treat households (close together) even when both are moderately far from vegetation. $\blacktriangledown$
\end{runexample}

\subsection{Network Formation and Pilot Sampling}

A second major point of departure this paper makes relative to the literature is in modeling the network underlying spillovers. In contrast to \cite{viviano2025policy}, the analysis here is not conditional on a maximum network degree. Instead, I assume that the network formation is through a dyadic graphon \citep{chandrasekhar,de2020econometric,lunde2023subsampling}.
Edge formation, therefore does not admit spillovers; i.e., the $i-j$ link does not directly influence the existence of the $k-l$ link. See \cite{chandrasekhar, de2020econometric} for an extensive discussion of the behavioral restrictions of this assumption. Despite its limitation, this dyadic formation assumption allows for tractability, particularly when identifying and estimating counterfactual treatment responses under interference facilitated by a network.
For each population $e\in\{0,1\}$, the node-level fundamentals are denoted by:
\begin{equation}
    \tilde F_i^{(e)}
     := 
    \left(U_i^{(e)},\X{e}_i,\Yc{e}_i\right),
    \qquad i=1,\ldots,\nee,
\end{equation}

\begin{assumption}[Graphon network formation with independent fundamentals]
\label{asmp:ub-fundamentals}\ 
\begin{enumerate}
    \item[a.] Node level fundamentals $\tilde F_i^{(e)}$ are independent and identically distributed, where $U_i^{(e)}\in[0,1]$ is an unobserved node characteristic.
    \item[b.] There exist a measurable symmetric function
$g:\Xc^2\times[0,1]^2\to[0,1]$, prescribed sparsity parameters
$0\leq\rho_\nee<1$, and, mutually independent
random variables
\[
    \xi_{ik}^{(e)}\sim\mathrm{Unif}[0,1],
    \qquad 1\leq i<k\leq\nee,
\]
independent of $(\tilde F_i^{(e)})_{i=1}^{\nee}$, such that
\begin{equation}\label{eq:pilot-graphon}
    A_{ik}^{(e)}
    =
    \one\!\left\{
        \xi_{ik}^{(e)}
        \leq
        \rho_\nee
        g\!\left(
            \X{e}_i,\X{e}_k,U_i^{(e)},U_k^{(e)}
        \right)
    \right\},
    \qquad
    A_{ki}^{(e)}=A_{ik}^{(e)},
    \quad A_{ii}^{(e)}=0.
\end{equation}
\end{enumerate}
\end{assumption}

Under Assumption~\ref{asmp:ub-fundamentals}, the vectors $\tilde F_i^{(e)}$ are IID. Conditional on fundamentals $\tilde F^{(e)}$, the network links form independently with the probabilities specified in Eqn.~\eqref{eq:pilot-graphon}. 
The parameter $\rho_\nee$ is an upper bound on link probabilities and need not imply a positive network density (e.g. $g:=0$ generates the empty graph). The statement $\nee\cdot \rho_\nee = O(1)$ controls the average degree of the resulting network \citep{lunde2023subsampling,graham2024sparse}.


\begin{runexample}
The contact network is defined by geographic proximity: the households are linked ($\bA{0}_{12}=1$) if their distance is below 250 meters so mosquitoes can fly between them. 
\begin{figure}[t] 
\centering 
\begin{subfigure}[b]{0.48\textwidth}
\centering
\includegraphics[scale=0.08]{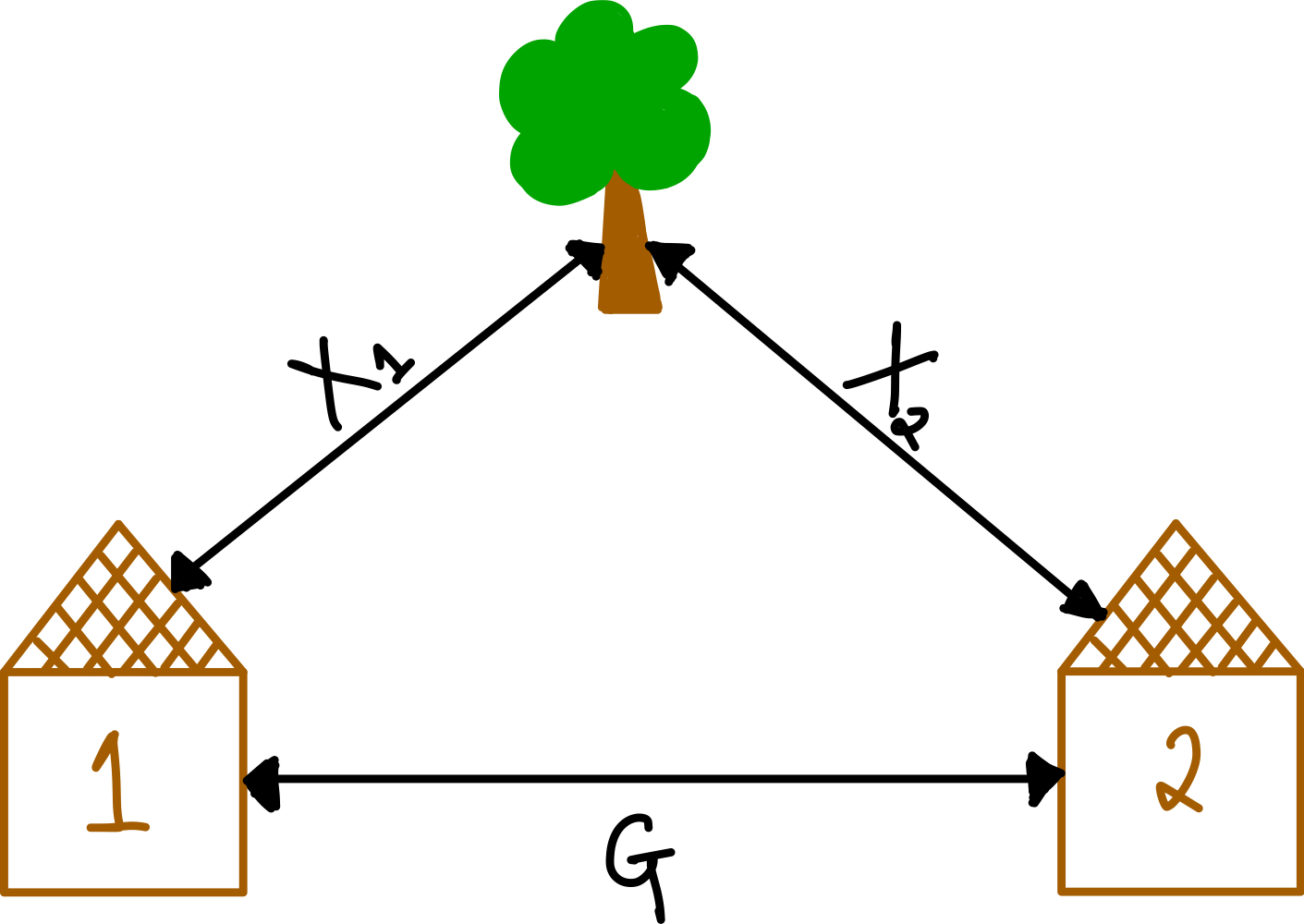} 
\caption{Target village in the running example.} 
\label{fig:target} 
\end{subfigure} \hfill
\begin{subfigure}[b]{0.48\textwidth} 
\centering 
\includegraphics[scale=0.15]{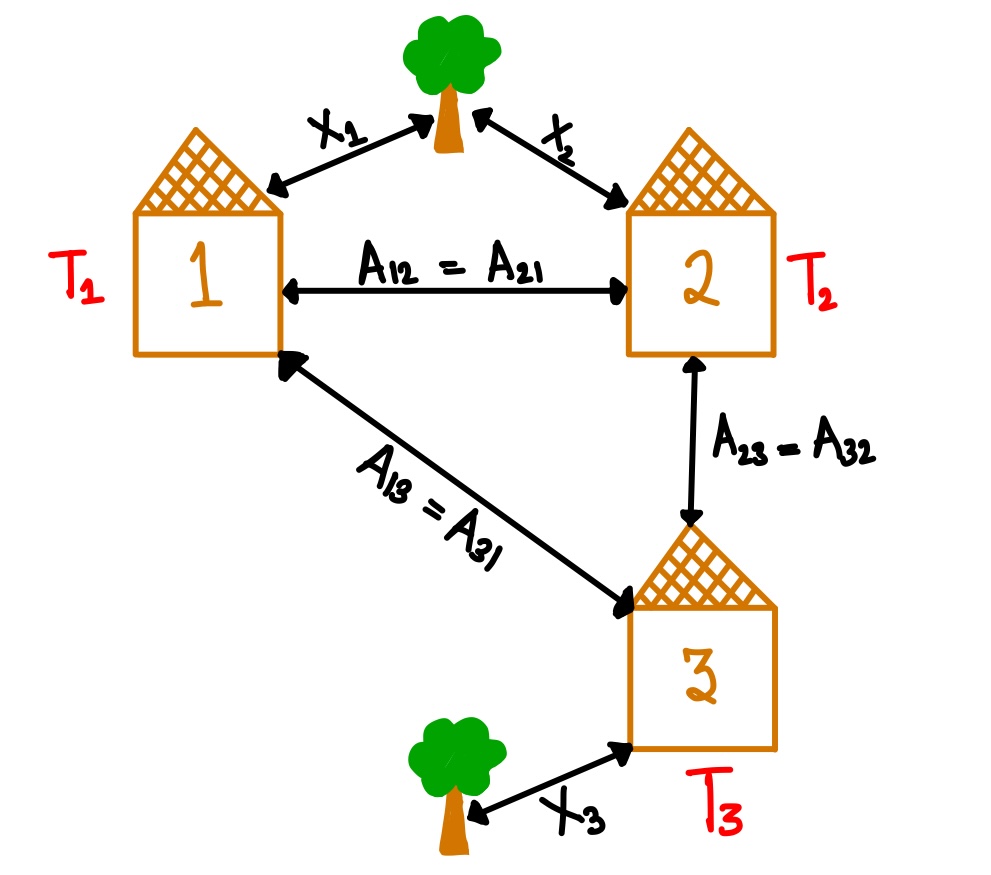} 
\caption{Pilot village in the running example.} 
\label{fig:pilot} 
\end{subfigure} 
\caption{Illustration of the target and pilot villages.} 
\label{fig:target_pilot} 
\end{figure}
This spatial network can be viewed as a latent-position graph in which each node has a latent location $\mathbf r_i\in\R^3$, and edge probabilities depend on distance. In the threshold case, this corresponds to the geometric graph model: 
\[ \bA{0}_{ij} = \one \{ \|\mathbf r_i-\mathbf r_j\|\leq250\}. \]
Let \(U_i \stackrel{iid}{\sim} \mathrm{Unif}[0,1]\), and let \(\phi:[0,1]\to\mathbb R^3\) be a measurable map such that \(\phi(U_i)\sim \mu\), the latent location distribution. Setting \(\mathbf r_i=\phi(U_i)\), the network formation model satisfies:
\[
\bA{0}_{ij}\mid U_i,U_j \sim \mathrm{Bernoulli}(g(U_i,U_j)),
\quad
g(u_i,u_j)=\mathbf 1\{\|\phi(u_i)-\phi(u_j)\|\le 250\}.\quad \blacktriangledown
\]
\end{runexample}

\begin{remark}
    The network formation model in this paper differs from the  framework of \citet[Section~3.4]{viviano2025policy}, whose regret analysis partitions observations into conditionally independent groups using dependency-graph coloring, with rates depending on the realized maximum degree and exposure overlap. 
Here, the dependence restrictions are imposed on the fundamentals of network formation rather than directly on the observed welfare contributions. Node-level fundamentals are independent across units, while links form independently conditional on observed and latent node characteristics. This permits homophilic sorting: conditioning on the realized network can induce dependence among outcome-relevant latent characteristics even for units that are distant or disconnected. Such dependence need not satisfy a finite-range conditional dependency-graph restriction. This distinction does not imply that dependency-graph methods exclude homophily generally, nor that the present model permits unrestricted homophily. The maintained local-response condition additionally requires the local state to summarize the conditional mean information about potential outcomes contained in the observed population configuration. Another important distinction is that the graphon model bounds expected degree, but does not impose a uniformly bounded maximum degree or exclude rare high-degree configurations. 
\end{remark}

\subsection{Experiment}

The planner observes the pilot dataset $\mathbf Z := (V^{(1)},\bT{1},\bY{1})$, 
where $\bY{1}$ is the vector of realized outcomes, $V^{(1)}$ is the observed pilot-population state. I use the following shorthands to denote the realized pilot exposure and the $i^{th}$ observation:
\begin{equation}\label{eq:pilot-E-Z}
 E^{(1)}_i := E_i^{(1)}(\bT{1};V^{(1)}),
 \qquad Z_i := (\Y{1}_i,L_i^{(1)},E^{(1)}_i).
\end{equation}

\begin{assumption}[No hidden treatments]\label{asmp2::consistent}
For each pilot unit $i$, the observed outcome equals the potential outcome
under its realized experimental exposure:
\[
 \Y{1}_i=\Y{1}_i(E^{(1)}_i)\qquad\text{almost surely}.
\]
\end{assumption}

This consistency restriction requires that the recorded assignment and network
correctly determine the exposure under which the outcome is observed. It
rules out unrecorded interventions or noncompliance that would alter that
exposure.

For subsequent arguments, denote the treatment augmented pilot fundamentals by
\[
    F_i^{(1)}
     := 
    \left(\tilde F_i^{(1)},T_i^{(1)}\right),
    \qquad
    F^{(1)}
     := 
    \left(F_1^{(1)},\ldots,F_\n^{(1)}\right).
\]

\begin{assumption}[Independent Bernoulli experimental design]
\label{asmp:experimental-design}
In the pilot population, experimental treatments satisfy \(T_i^{(1)}
    \overset{\mathrm{i.i.d.}}{\sim}
    \mathrm{Bernoulli}(p_\n),\)
for some $p_\n$ known to the planner.
The treatment vector $\bT{1}$ is independent of the collection
\(
    \left(
        (\tilde F_i^{(1)})_{i=1}^{\n},
        (\xi_{ik}^{(1)})_{1\leq i<k\leq\n}
    \right).\)
\end{assumption}
Assumption~\ref{asmp:experimental-design} together with Assumption~\ref{asmp:ub-fundamentals} implies that the vectors $F_i^{(1)}$ are IID.
The independent assignment in Assumption~\ref{asmp:experimental-design} also implies generalized unconfoundedness \(\Yc{1}_i\independent E^{(1)}_i\mid V^{(1)}\).
The remaining restrictions concern the adequacy of the local state $L$.

\begin{assumption}[Local information restrictions]
\label{asmp2::unconfoundedness}
\begin{enumerate}
 \item[a.] In population $e \in \{0,1\}$, for every unit $j$ and exposure
 $c\in\cC^{(e)}$,
 \(
 \ex[\Y{e}_j(c)\mid V^{(e)}]
 =\ex[\Y{e}_j(c)\mid L_j^{(e)}]
 \text{\ almost surely}
. \)
 \item[b.] In the pilot population, exposure irrelevance holds for every unit:
 \(
 E^{(1)}_i\independent V^{(1)}\mid L_i^{(1)}.
 \)
\end{enumerate}
\end{assumption}

Part a is a conditional-mean, no-contextual-effects restriction: once the local
state is fixed, the remaining target configuration does not change the mean
response at a feasible exposure \citep{manski1993reflection}. Part b requires that the
pilot exposure law depends on the full observed configuration only through
$L_i^{(1)}$. Under LSE and independent Bernoulli assignment, this law depends
on the degree encoded in $C_i^{(1)}$, so part b holds.
Combining generalized unconfoundedness with Assumption~\ref{asmp2::unconfoundedness}.b gives \citep{constantinou2017extended}: 
\begin{equation}\label{eqn2::strict_unconfoundedness}
 \Yc{1}_i ~ \independent ~ E^{(1)}_i ~\mid ~ L_i^{(1)}.
\end{equation}

\begin{definition}[Generalized Propensity Score]\label{def::gen_prop_score}
For $\ell=(x,C)$, the generalized propensity for pilot unit $i$ and feasible exposure $c\in C$ is $ e_{i,c}(\ell) := \p(E^{(1)}_i=c\mid L_i^{(1)}=\ell).$
\end{definition}

By exposure irrelevance, this equals the conditional exposure probability
given the full observed pilot state and can be computed
from the known experiment design. 

\subsection{Identification}

The mean treatment response function measures the expected outcome for a unit given its own local state and treatment exposure. I now formalize assumptions under which pilot data identifies this function at exposures that are relevant to the target population welfare.

\begin{definition}[Mean Treatment Response Function]
\label{def:mean-response}
For any $\ell=(x,C)$ and $c\in C$:
\begin{equation}
 \mu_c^{(e)}(\ell) := \ex[\Y{e}_i(c)\mid L_i^{(e)}=\ell].
\end{equation}
\end{definition}

The above definition implicitly assumes that the mean treatment response functions do not vary across units. This is formalized in Assumption~\ref{asmp2::identical} which imposes identical distributions of local-states and mean treatment responses across units within either population. 
\begin{assumption}[Identically distributed units]
\label{asmp2::identical}
In any population $e\in\{0,1\}$:
\begin{enumerate}
 \item[a.] The law of local state $L_i^{(e)}$, denoted by $\p_L^{(e)}$, is the same for every unit.
 \item[b.] For every $c\in\cC^{(e)}$ and all units $i,j$, $\ex[\Y{e}_i(c)\mid L_i^{(e)}=\ell]
 =\ex[\Y{e}_j(c)\mid L_j^{(e)}=\ell]$ 
 for almost every $\ell=(x,C)$ with $c\in C$.
\end{enumerate}
\end{assumption}

Learning a target response additionally requires that target-relevant local
states and exposures can be represented in the pilot. The following site
support restriction extends the usual covariate-support condition to my setting
\citep{hotz2006evaluating}.

\begin{assumption}[Site support restrictions]\label{asmp2::support} \
\begin{enumerate}
 \item[a.] The local-state laws satisfy $\p_L^{(0)}\ll \p_L^{(1)}$ on
 the common space $\mathcal L$.
 \item[b.] Their Radon--Nikodym derivative, $s(\ell) := d\p_L^{(0)}/d\p_L^{(1)}$, is known and satisfies $0\leq s(\ell) \leq\kappa$, $\p_L^{(1)}$-almost surely,
 for a prescribed finite $\kappa\geq1$.
\end{enumerate}
\end{assumption}

The Radon--Nikodym derivative is unit-invariant by Assumption~\ref{asmp2::identical}.a. The
change-of-measure formula concerns two marginal laws. Thus every exposure feasible on a
positive-probability target local state is represented in the pilot. Treating $s$ as known
avoids an additional nuisance-estimation error; extending the analysis to an
estimated ratio requires controlling that error
\citep{van2007empirical,hirano2003efficient}. 


\begin{assumption}[Stability]\label{asmp2::stability}
For $\p_L^{(0)}$-almost every local state $\ell=(x,C)$ with $c\in C$,
\[
 \mu_c^{(0)}(\ell)=\mu_c^{(1)}(\ell).
\]
\end{assumption}

Stability compares mean responses at the \emph{same} local state. Together
with site support, it permits transporting pilot responses without requiring
the two populations' covariates or networks to have identical distributions.

The welfare induced by any treatment rule is a weighted average of the mean treatment responses of units in the target population, evaluated at the local state of each unit for exposures induced by the rule (i.e., policy-relevant exposures). Definition~\ref{defn2::exp_freq_pi} formalizes the weights attached to any local state--exposure pair as a network exposure kernel.  

\begin{definition}[Network exposure kernel]\label{defn2::exp_freq_pi}
For any rule $\pi \in \Pi$, denote the weight assigned to local state $\ell=(x,C)$ and exposure $c\in C$ for unit $j$ by:
\begin{equation}
 q_j(\ell,c;\pi)
 := \p(E_j^\pi=c\mid L_j^{(0)}=\ell), 
 \qquad \overline q(\ell,c;\pi)
  := \tfrac1\J\sum_{j=1}^{\J}q_j(\ell,c;\pi).
\label{eq:ub-common-kernel}
\end{equation}
\end{definition}


\begin{definition}[Policy-relevant exposures]
\label{def:policy-relevant-exposures}
For a fixed policy $\pi\in\Pi$, define
\begin{equation}\label{eq:policy-relevant-pairs}
 \mathcal E_\pi(\p)
 :=\{(\ell,c):\ell=(x,C),\ c\in C\cap\cC^{(0)},\
                        s(\ell)\overline q(\ell,c;\pi)>0\}.
\end{equation}
Overlap is imposed separately for every $\pi\in\Pi$ with the same
numerical lower bound $\gamma_\J$.
\end{definition}

\begin{assumption}[Policy-relevant overlap]\label{asmp2::overlap}
There exists $\gamma_\J>0$, such that for every pilot unit $i$, and every
$\pi\in\Pi$, and every $c\in\cC^{(0)}$, propensity scores are known and satisfy:
\begin{equation}\label{eq:ub-relevant-overlap}
 e_{i,c}(\ell)\geq\gamma_\J
 \quad\text{for }\p_L^{(1)}\text{-almost every }\ell
       \text{ such that }(\ell,c)\in\mathcal E_\pi(\p).
\end{equation}
\end{assumption}

The overlap bound is indexed by target size. It suffices for the welfare
criterion, which gives zero weight to policy-irrelevant pairs. Restricting
attention to these pairs does not, by itself, yield a target-size-only
bound for every exposure mapping or for a treatment probability approaching
zero or one. The following LSE calculation gives that bound when the
experimental treatment and control probabilities are uniformly positive.

\begin{runexample}\label{ex:prop_score}
In the pilot village, let
$\operatorname{deg}_i^{(1)} := \sum_{j\neq i}A_{ij}^{(1)}$
denote household $i$'s degree. Under LSE, the realized exposure and
feasible exposure set are
\[
 E^{(1)}_i=\left(T_i^{(1)},\sum_{j\neq i}A_{ij}^{(1)}T_j^{(1)},
                          \operatorname{deg}_i^{(1)}\right),
 \qquad C(d) := \{0,1\}\times\{0,\ldots,d\}\times\{d\}.
\]
For $\ell=(x,C(d))$ and $c=(a,b,d)\in C(d)$, independent Bernoulli
assignment gives
\[
 e_{i,c}(\ell)
 =p_\n^a(1-p_\n)^{1-a}\binom dbp_\n^b(1-p_\n)^{d-b}
 \geq\underline p_\n^{\,d+1},
 \qquad \underline p_\n := \min\{p_\n,1-p_\n\}.
\]
This law depends on the full pilot state only through the degree recorded
in $L_i^{(1)}$, verifying exposure irrelevance in this example. If both
populations use LSE, a target node has degree at most $\J-1$; site support
therefore implies $s(x,C(d))=0$ for pilot-almost every state with
$d>\J-1$. Consequently, at every policy-relevant pair $ e_{i,c}(\ell)\geq\underline p_\n^{\,d+1}
                    \geq\underline p_\n^{\,\J}.$
On any model class with $p_\n\in[p_0,1-p_0]$ for a fixed
$0<p_0\leq1/2$, one may take $\gamma_\J=p_0^{\J}$. For a fixed $p_0$,
the bound depends on $\J$, not on the pilot size $\n$. This does not bound
all pilot degrees or all pilot exposure probabilities. For $\J=2$ and
$p_\n=1/2$, for example, the bound on relevant propensities is $1/4$
even when the pilot population is much larger.
$\blacktriangledown$
\end{runexample}

\begin{remark}
    Under LSE, Example~\ref{ex:prop_score} bounds relevant propensities in terms of target size. When the target population grows, even relevant propensities can become small. 
This motivates the cutoff empirical welfare extension in Appendix~\ref{app:cutoff-welfare}, related to the trimming extension in \citet[Section~4.1]{viviano2025policy}.
\end{remark} 

\begin{assumption}[Known averaged target exposure kernels]
\label{asmp2::known_marginals}
The planner knows $\overline q(\ell,c;\pi)$ for each $\pi\in\Pi$ and
$c\in\cC^{(0)}$, at local states where $s(\ell)>0$.
\end{assumption}

Knowledge of the target-state law $P_V$ is sufficient to construct the kernels. Estimating these kernels would introduce an additional approximation term.

\begin{runexample}
For two-household target villages, knowledge of the distribution of $V$---the
household covariates and their contact network---permits conditional draws
of villages given $L_j^{(0)}$ and hence evaluation of each $q_j$ and
$\overline q$. Administrative records, censuses, or pre-campaign surveys may
supply this information. $\blacktriangledown$
\end{runexample}

\begin{theorem}[Identification of Welfare]\label{lem::iden_welfare}
Suppose Assumptions~\ref{asmp2::eff_treat}, \ref{asmp2::pot_out},
\ref{asmp:ub-fundamentals}, \ref{asmp2::consistent}, \ref{asmp:experimental-design},
\ref{asmp2::unconfoundedness}, \ref{asmp2::identical},
\ref{asmp2::support}, \ref{asmp2::stability},
\ref{asmp2::overlap}, and \ref{asmp2::known_marginals} hold.
Then, for every fixed $\pi\in\Pi$ and every pilot unit $i$,
\begin{align}
 W(\pi;\p) & ~=~ \ex_\p[\Upsilon_i(Z_i;\pi)] ~=~ \ex_\p\left[\tfrac{1}{\n}\sum_{i=1}^\n \Upsilon_i(Z_i;\pi)\right]\ ,  \nonumber \\
  \text{where} \quad \Upsilon_i(Z_i;\pi)
 & ~:=~ \sum_{c\in\cC^{(0)}}\tfrac{\Y{1}_i\one\{E^{(1)}_i=c\}}{e_{i,c}(L_i^{(1)})}
                     s(L_i^{(1)})
              \overline q(L_i^{(1)};c,\pi) \label{eq:ub-ipw-weight}
\end{align}
\end{theorem}

\section{Welfare Estimation and Policy Learning}
\label{subsec:aipw-welfare}

In the absence of knowledge of the DGP, the planner uses data $\mathbf Z$ to construct a data-dependent allocation rule that approximates the welfare maximizer.
 I formalize such data-dependent procedures as \emph{statistical treatment rules}, i.e. jointly measurable maps
\(
    (\mathbf z,v)\longmapsto
    \pi_{\mathbf z}(v)\in\Tc.
\) where $\pi \in \Pi$.
Following convention, I define uniform regret as a measure of how accurately a statistical treatment rule approximates the welfare generated by an optimal treatment allocation rule when $\p$ is known.

\begin{definition}[Uniform regret]\label{def:uniform-regret} Let $\Pc$ be the nonempty class of pilot
laws and target evaluation laws satisfying Assumptions~\ref{asmp2::eff_treat}--\ref{asmp2::known_marginals}. The target population drawn is independent of the pilot experiment. For a statistical rule, define
    \begin{equation}
        R(\tilde \pi) = \sup_{\p\in \Pc}~\ex \left[~\sup_{\pi \in \Pi} W(\pi;\p)-W(\tilde \pi;\p) \right].
    \end{equation}
\end{definition}

\subsection{Empirical Welfare Maximization}

I now extend empirical welfare maximization to the global treatment allocation problem. I start by defining the empirical welfare function.  

\begin{definition}[Empirical welfare]\label{def:emp_wel}
For any $\pi\in\Pi$, define $\widehat\W_{\n}(\pi)
  := \frac1\n\sum_{i=1}^{\n}\Upsilon_i(Z_i;\pi),$
using the same welfare contribution defined in Theorem~\ref{lem::iden_welfare}.
\end{definition}

The treatment rule that maximizes the empirical welfare is denoted by the empirical welfare maximizer. 

\begin{definition}[Empirical welfare maximization]
\label{def:EWM}
An empirical welfare maximizer is a statistical treatment rule
$(\mathbf z,v)\mapsto\widehat\pi(v)$ that is jointly
measurable, has $\widehat\pi\in\Pi$, and satisfies
\[
 \widehat W_\n(\widehat\pi)
 =
 \sup_{\pi\in\Pi}\widehat W_\n(\pi)
 \qquad
 \p\text{-almost surely for every }\p\in\Pc.
\]
I assume that such a jointly measurable rule exists and attains the supremum.
\end{definition}

To derive uniform regret bounds for the empirical welfare maximizer, I study
the empirical process indexed by the policy class.
The welfare-contribution class, denote by $\mathcal F_\Pi
  := \{(i,z)\mapsto\Upsilon_i(z;\pi):\pi\in\Pi\}$ must be of restricted complexity. 

\begin{assumption}[Finite Natarajan dimension and measurability]\label{asmp:Nat}
The nonempty class $\Pi$ has finite Natarajan dimension
$1 \leq N_{dim}<\infty$. There is a countable subclass
$\Pi_0\subseteq\Pi$ such that every $\pi\in\Pi$ is the pointwise limit of a
sequence from $\Pi_0$, with the allocation set equipped with the discrete
topology. 
\end{assumption}

The countable approximation makes the suprema below measurable. An equivalent
separability condition or an outer-expectation formulation can be used instead.
Finite Natarajan dimension controls the complexity of the class of allocation rules considered and is a natural analogue of the VC dimension is multiclass learning problems. If $N_{dim}=0$, all policies in
$\Pi$ coincide pointwise and an empirical maximizer in $\Pi$ has zero regret.

\begin{remark}
\label{rem:Nat-target-size}
Assumption~\ref{asmp:Nat} concerns the chosen policy class at each
target size $\J$ and budget $B$; it does not require $N_{dim}$ to
remain uniformly bounded as $\J$ grows. The dependence of $\Pi$
and $N_{dim}$ on $(\J,B)$ is suppressed in the notation. Although
the number of feasible allocations
$\Q=\sum_{b=0}^{B}\binom{\J}{b}$ grows with $\J$ for a fixed
positive budget, this does not by itself imply that the Natarajan
dimension increases. Its behavior depends on how the policy
classes are specified across population sizes, and there is no
general monotonicity without additional structure. In particular,
the shared-score class below has dimension bounded independently
of $\J$ when its number of features is fixed. Conversely, allowing
additional features or allocation-specific coefficients can increase
policy complexity. The regret bounds distinguish this complexity
from the size of the allocation space through the separate
appearance of $N_{dim}$ and $\Q$ in
Eqn.~\eqref{eq:ub-complexity}.
\end{remark}

\begin{example}[Finite-dimensional global allocation rules]
\label{ex:finite-dim-global-rules}
Fix a known measurable feature map
$\phi:\mathcal L\times\cC^{(0)}\to\R^{d_\phi}$,
$1\leq d_\phi<\infty$, and a nonempty countable coefficient set
$\Theta\subseteq\R^{d_\phi}$ (e.g. rationals in $\R^{d_\phi}$). Define a joint allocator based on a linear score:
\begin{align}
    \Phi(v,\mathbf t)
    &:=\frac1\J\sum_{j=1}^{\J}
      \phi\!\left(
          L_j^{(0)}(v),E_j^{(0)}(\mathbf t;v)
      \right),
      \label{eq:global-linear-features}\\
    \pi_\theta(v)
    &\in\argmax_{\mathbf t\in\Tc}
       \theta^\top\Phi(v,\mathbf t),
    \qquad
    \Pi_\phi:=\{\pi_\theta:\theta\in\Theta\},
      \label{eq:global-linear-policy}
\end{align}
selecting the first maximizer in a fixed ordering of $\Tc$.
Each rule is measurable and budget feasible, and its assignments
are chosen jointly using the population state. The feature-difference
argument of \citet[Theorem~7]{daniely2014optimal} gives
\begin{equation}\label{eq:global-linear-Natarajan}
    \dim_N(\Pi_\phi)\leq d_\phi.
\end{equation}
Since $\Pi_\phi$ is countable, Assumption~\ref{asmp:Nat} holds
with $\Pi_0=\Pi_\phi$.

Under LSE in Eqn.~\eqref{eg::lse}, a concrete choice of index score is
\begin{equation}\label{eq:global-linear-LSE-features}
    \phi\bigl((x,C(d)),(a,s,d)\bigr)
    :=
    \left(
        a,\;ax^\top,\;
        \frac{s}{\max\{d,1\}},\;
        \frac{as}{\max\{d,1\}}
    \right)^\top,
    \qquad d_\phi=d_x+3,
\end{equation}
on feasible pairs, with zero extension elsewhere and
$C(d)=\{0,1\}\times\{0,\ldots,d\}\times\{d\}$.
The score incorporates covariate-dependent direct returns,
neighborhood treatment, and their interaction; it does not impose
linearity on the true mean response. For fixed $d_x$, the dimension
bound is independent of $\J$, even though the feasible-allocation
space grows.
\end{example}

\section{Finite-Sample Regret Bounds for Global Rules}
\label{sec:regret-bounds}

This section establishes an upper bound for empirical welfare maximization, followed by a
minimax lower bound. The two results are compared on the same maintained model class; except that the lower bound adds the welfare-relevant richness assumption.


\subsection{Upper bound on regret}
\label{subsec:upper}

The upper bound quantifies how closely empirical welfare maximization approximates oracle welfare.
Define the finite catalog of possible pilot feasible-exposure sets by:
\begin{equation}
 \mathscr C_\n
 := \{C_i^{(1)}(A,\mathbf x):i\in[\n],\ A\in\sA_\n,
                                  \ \mathbf x\in\Xc^\n\},\quad 
 K_{\mathrm{set},\n} := |\mathscr C_\n|.
 \label{eq:ub-feasible-catalogue}
\end{equation}

Since the exposure space is finite, $1\leq K_{\mathrm{set},\n}\leq2^{|\cC^{(1)}|}-1$.
Under LSE, $K_{\mathrm{set},\n}\leq\n$, because each feasible set is indexed by degree.
There is no finite-support restriction on the covariates or target states.
Set
\begin{equation}\label{eq:ub-complexity}
 a_\Pi := 1+N_{dim}\log(12e\Q^2),
 \qquad
 \vartheta_\n := 1+N_{dim}\log(12e\Q^2\n^2 K_{\mathrm{set},\n}).
\end{equation}
The $e$ in this display is Euler's number. The additional logarithm accounts for
kernel evaluations at local states generated by network and node resampling.
The bounds $M,\gamma_\J,\kappa,\rho_\n,d_\n^{\mathrm{sc}}$ and the
policy class are common across $\Pc$. 

\begin{theorem}[Upper bound for regular empirical welfare]\label{thm:upper-new}
Let $\Pc$ be the common model class in Definition~\ref{def:uniform-regret},
and suppose Assumption~\ref{asmp:Nat} holds, then there exists a universal constant $C<\infty$ such that every empirical welfare maximizer satisfies
\begin{equation}\label{eq:upper-new}
\begin{split}
 R(\widehat\pi)
 &\leq 2M\wedge\Bigg\{CM\kappa\Bigg[
  \sqrt{\frac{a_\Pi}{\n\gamma_\J}} + \frac{d_\n^{\mathrm{sc}}}{\gamma_\J}
    \left(
       \sqrt{\rho_\n\vartheta_\n}
       +\n\rho_\n\sqrt{\frac{\vartheta_\n}{\n}}
       +\frac{\vartheta_\n}{\n}
    \right)
 \Bigg]\Bigg\}.
\end{split}
\end{equation}
If $\rho_\n=0$ (empty graph) or
$d_\n^{\mathrm{sc}}=0$ (SUTVA holds), the familiar bound appears.
\begin{equation}\label{eq:upper-new}
\begin{split}
 R(\widehat\pi)
 &\leq 2M\wedge\Bigg\{CM\kappa 
  \sqrt{\frac{a_\Pi}{\n\gamma_\J}} \Bigg\}.
\end{split}
\end{equation}
\end{theorem}


\begin{cor}[Sparse local-score case]\label{cor:sparse-bdd-degree}
Suppose the conditions of Theorem~\ref{thm:upper-new} hold uniformly along a
sequence of models, with
\[
 d_\n^{\mathrm{sc}}=O(1),\qquad \rho_\n=O(\n^{-1}),\qquad \kappa=O(1).
\]
Then, for a constant $C$ depending only on these uniform bounds,
\[
 R(\widehat\pi)
 \leq 2M\wedge\left\{CM\left[
       \sqrt{\frac{\vartheta_\n}{\n}}
       +\frac{\vartheta_\n}{\n}\right]\right\}.
\]
\end{cor}

The sparsity restriction bounds expected degree, not maximum degree.
Overlap is now required only on policy-relevant pairs. Under LSE and a
fixed lower bound $p_0$ on treatment and control probabilities,
Example~\ref{ex:prop_score} permits $\gamma_\J=p_0^{\J}$.
Thus fixed target size gives a pilot-size-independent overlap bound without
a maximum-degree restriction. When $\J$ grows, this bound can still
be small. Appendix~\ref{app:cutoff-welfare} develops a separate cutoff
empirical welfare criterion and upper bound, with explicit conditions
under which omitted target welfare is negligible relative to estimation
error.

\subsection{Lower bound on regret}
\label{subsec:lower}

I now provide a lower bound on a similar model class $\Pc$ as in
Theorem~\ref{thm:upper-new}. Finite Natarajan dimension does not itself
ensure independently variable welfare trade-offs. Allocations can change
individual exposures without changing aggregate welfare. The only additional assumption therefore is a welfare-relevant richness condition on the policy class.
For $e\in\{0,1\}$, let
\(
 v_\varnothing^{(e)}(\mathbf x)
  := V^{(e)}(\mathbf0_{\nee\times\nee},\mathbf x)
\)
be the population state generated by the maintained feature transformation
at an empty graph which only admits exposure labels $c_{\varnothing,0}$ and
$c_{\varnothing,1}$ by Assumption~\ref{asmp:ub-locality}.

\begin{assumption}[Richness of $\Pi$]
\label{asmp:welfare_shatter}
There exist $m\geq1$ distinct active covariate values $x_1,\ldots,x_m$
and a further value $x_0$ in $\Xc$. Write
$\Xc_m := \{x_0,x_1,\ldots,x_m\}$ and $N_m := \{-1,1\}^m$.
For each $\nu\in N_m$, the class $\Pi$ contains a
 feasible rule $\pi_\nu$ whose behavior on every empty graph
 $v_\varnothing^{(0)}(\mathbf x)$, $\mathbf x\in\Xc_m^{\J}$, is as
 follows. Define
 \[
 I_\nu(\mathbf x) := \{j:\mathbf x_j=x_r\text{ for some }r\geq1
                                      \text{ with }\nu_r=1\}.
 \]
 The rule treats only the first $\min\{B,|I_\nu(\mathbf x)|\}$ units in
 this set, ordered by their indices. 
\end{assumption}

The assumption specifies behavior on every configuration generated by the construction,
not only on a collection of anchors. Let
$\mathbf x_r := (x_r,\ldots,x_r)$ and
$v_r := v_\varnothing^{(0)}(\mathbf x_r)$.
At $v_r$, the two possible allocations are
$\mathbf t_B := (1,\ldots,1,0,\ldots,0)$, with $B$ leading ones, and
$\mathbf0$. Thus $\{v_1,\ldots,v_m\}$ is Natarajan-shattered and
$m\leq N_{dim}$. When $B<\J$, this condition is not automatic for a
deterministic permutation-equivariant class: indistinguishable units in a
homogeneous edgeless state cannot be split by such a rule.

\begin{theorem}[Lower bound]\label{thm:lower-new}
Let $\Pc$ be the same common model class as in
Theorem~\ref{thm:upper-new}, including Assumption~\ref{asmp:Nat}.
Suppose additionally that Assumption~\ref{asmp:welfare_shatter} holds.
If $1\leq B\leq\J$, $\n\geq1$, $M>0$, and $0<\gamma_\J<1/2$, then
\begin{equation}\label{eq:lower-new}
 \inf_{\widetilde\pi}\sup_{\p\in\Pc}\ex_\p
 \left[\sup_{\pi\in\Pi}W(\pi;\p)-W(\widetilde\pi;\p)\right]
 \geq\frac{M}{128}
 \left[\frac{B}{\J}\wedge
             \sqrt{\frac{Bm}{\J\n\gamma_\J}}\right].
\end{equation}
\end{theorem}


The theorem concerns the model class defined above. The proof chooses $p_\n$ of order $\gamma_\J$.
It does not claim the same overlap-dependent lower bound for every fixed
experimental design. The empty-network submodel does not establish sharpness
of the upper bound's network terms, nor for the budget capacity.

\begin{remark}[Comparison of upper and lower bounds]
The lower bound additionally requires
Assumption~\ref{asmp:welfare_shatter}.
Suppose $B/\J\geq b_0>0$, $m\geq a_0N_{dim}$ for $a_0>0$, and the
conditions of Corollary~\ref{cor:sparse-bdd-degree} hold with
$\vartheta_\n\leq\n$. Then
\begin{align*}
 cM\left[1\wedge\sqrt{\frac{N_{dim}}{\n}}\right]
 &\leq\inf_{\widetilde\pi}\sup_{\p\in\Pc}\ex_\p
       \left[\sup_{\pi\in\Pi}W(\pi;\p)-W(\widetilde\pi;\p)\right]\\
 &\leq R(\widehat\pi)
 \leq\eta_\n+CM\sqrt{\frac{\vartheta_\n}{\n}}.
\end{align*}
The constants may depend on the common budget fraction, richness fraction,
and bounds in the corollary. If $K_{\mathrm{set},\n}$ is polynomial in $\n$, the
bounds match up to logarithmic factors. This statement does not apply to an
arbitrary $\Pi$ or a fixed nonempty network submodel.

\end{remark}

\section{Global versus Local Treatment Rules} \label{sec:compare}

The treatment allocation problem described in this paper extends the local problem typically studied in the literature \citep{kittet18, viviano2025policy}. 
The local problem selects a treatment rule that assigns treatments on a unit-by-unit basis, while still aiming to maximize the same interference-aware welfare defined in this paper (Eqn.~\ref{eq2:welfare}). A local rule (or individualistic rule) is a mapping $\pil \in \Lambda := \{ \Xc \times \Gc{0} \to \{0,1\}\}$ which assigns treatment based on the observed local information $(X_j^{(0)},G_j^{(0)})$.
In contrast, I refer to the treatment allocation problem in this paper (Eqn.~\ref{eq2:planner}) as the global problem, and the associated treatment rule $\pig: \mathcal V\to\Tc$ as the global rule.
In order to compare the performance of local and global rules in a unified framework, I first define the global allocation induced by a local rule, denoted by $\pigpil$. 
This is defined below as:
\begin{definition}[Global allocation induced by local rule]
    For any local rule $\pil$, denote by $\pigpil(V) = \bigl[\pil(\X{0}_j,\G{0}_j)\bigr]_{j=1}^{\J}$ the global rule induced by it. 
\end{definition}
By construction, global allocations induced by local rules are restricted to satisfy a \emph{symmetry property}: if two units have identical covariates and network features (local information), they receive the same treatment assignment.
Formally, for any fixed state $v$, for any $j\neq k$, if $(X_j^{(0)}(v),G_j^{(0)}(v))=(X_k^{(0)}(v),G_k^{(0)}(v))$, then $\pil(X_j^{(0)}(v),G_j^{(0)}(v))=\pil(X_k^{(0)}(v),G_k^{(0)}(v))$.
Consequently, no local rule can induce an allocation $\bt{0}\in\Tc$ that treats two units differently when their observed local information coincides. 
Thus, the local problem solves for the optimal rule subject to this symmetry constraint:
\begin{equation}\label{eq:local}
  {\pil}^*_\p = \underset{\pil\in\Lambda_B(\p)}{\arg\sup}\ \W(\pigpil;\p), \quad \text{where } \Lambda_B(\p) := \{\pil \in \Lambda \text{ such that } \pigpil(V)\in\Tc~-\p\text{ a.s. }\}
\end{equation}
where $\Lambda_B(\p)$ denotes the class of all feasible local rules. Notice here that care is taken to ensure that the global allocations induced by local rules satisfy the budget constraint.  Since the class of all global rules contains those induced from local rules, the optimal global rule attains weakly larger welfare than the optimal local rule. I define this difference in welfare from the two optimal rules as the welfare gap and denote it by:
\begin{equation}
    \Delta_\p ~=~ \W(\pig^*_\p;\p) - \W\bigl(\pig_{\pil^*_\p};\p\bigr) ~ \geq ~ 0.
\end{equation}

It is crucial to note that this inequality holds even when both global and local rules have access to the same covariate and network information. Contrary to the suggestion in the previous version of this paper \citep[Chapter~1]{ananth2021decision} and in \cite[Remark 4]{viviano2025policy}, the inequality does not require the full network to be observed for the global rule.
The inequality arises because the symmetry constraint may be binding in the local problem. 
When the marginal welfare from treating nodes net of the opportunity cost of satisfying the capacity constraint is positive under some exposures and negative under others, the optimal allocation may require coordinating treatment decisions across individuals, for example treating one individual while withholding treatment from another with identical local information in order to induce favorable exposure conditions overall. Global rules can implement such joint, asymmetric allocations, whereas local rules generally cannot.

Example \ref{ex:full_net} illustrates this with a simple, known DGP and fixed target-population state to show that when the planner uses identical information to construct both local and global rules, the symmetry property can result in a welfare gap, even without capacity constraints.

\begin{example}[Full network information]\label{ex:full_net}
    Suppose the planner can optimally allocate $\B=4$ treatments in a target population comprising $\J = 4$ units. Given the constraint, the feasible set of allocations is given by $\Tc = \left\{ 0,1 \right\}^4$.
    For ease of exposition, I assume that the known marginal distribution of covariates and the network is degenerate. Specifically, I assume that the support of covariates is $\supp{X_j} = \{1\}$ and the support over networks is concentrated on a network with adjacency matrix $A_{jk} = \mathbf{1}\{ j+k = 2m-1 \text{ where } m \in \mathbb{N}\}$. This adjacency matrix corresponds to a 4-node box as illustrated in Figure~\ref{fig:ex2}.
    \begin{figure}
        \centering
        \begin{tikzpicture}[scale=0.5]
            \node[circle, draw, scale=0.7] (1) at (0,2) {1};
            \node[circle, draw, scale=0.7] (2) at (2, 2) {2};
            \node[circle, draw, scale=0.7] (3) at (2,0) {3};
            \node[circle, draw, scale=0.7] (4) at (0,0) {4};
            \draw (1) -- (2);
            \draw (2) -- (3);
            \draw (3) -- (4);
            \draw (4) -- (1);
            \end{tikzpicture}
        \caption{Network in Example \ref{ex:full_net}}
        \label{fig:ex2}
    \end{figure}
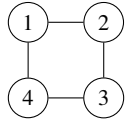
    In this example, I retain the LSE spillovers assumption, allow for the mean treatment response under different exposures to fluctuate in sign:
    \begin{equation}
        \ex\left[Y_j(E_j(\mathbf t)) \bigg| L_j \right] = \left(-t_j\right)^{\left(1 + \sum_i A_{ji} t_i \right)}
    \end{equation}
    
    Suppose a planner only observes each units covariate $X_j$ and list of each units neighbors. This satisfies Assumption~\ref{asmp2::pot_out}. 
    Define welfare as in \ref{eq2:welfare} and denote by ${\pig}^*$ the welfare maximizing global treatment rule and by ${\pil}^*$ the welfare maximizing local treatment rule. It can be seen that the optimal global rules assign treatments to any one connected pair of nodes, i.e., $\{\{1,2\}, \{2,3\}, \{3,4\}, \{4,1\}\}$. Notice that the optimal allocation breaks the symmetry property. For example if nodes $\{1,2\}$ are the only treated, then node 1, with local information (1,{2,4}), is treated, whereas node 3, with the same local information, is not. The optimal local rule must satisfy symmetry. In this example, it can be seen that the optimal local rule treats none. 
    $\blacktriangledown$
\end{example}


\subsection{Characterizing the welfare gap}
\label{subsec:gap-spillover-extension}
I now present a result to put a lower bound on the welfare gap. This result is on the population oracle comparison of the welfare gap, not statistical regret or a comparison with an arbitrarily restricted learning class $\Pi$.
For each target state $v$, define conditional welfare as the conditional
expected average outcome given global state $V=v$, denoted by:
\begin{equation}\label{eq:gap-conditional-response}
 w_\p(v,\mathbf t)
  := \frac1\J\sum_{j=1}^{\J}
       \mu^{(0)}_{E_j^{(0)}(\mathbf t;v)}(L_j^{(0)}(v)),
 \qquad \mathbf t\in\Tc.
\end{equation}
The absence of complexity restrictions implies that any optimal allocation is also a statewise maximizer,
\(
 \sup_{\pi}W(\pi;\p)
 =\ex_\p\!\left[\max_{\mathbf t\in\Tc}w_\p(V,\mathbf t)\right]
\).
For $1\leq B\leq\J$ and $a\in\{0,1\}$:
\begin{equation}
 D_j(v) := w^*_{j,1}(v)-w^*_{j,0}(v), \quad  \text{where} \quad w^*_{j,a}(v) := \max_{\mathbf t\in\Tc:\,t_j=a}w_\p(v, \mathbf t).
                  \label{eq:gap-inclusion-advantage}
\end{equation}
The quantity $D_j(v)$ compares the best feasible allocations with and
without $j$, allowing all other assignments to be reoptimized. Every
optimal allocation includes $j$ when $D_j(v)>0$ and excludes it when
$D_j(v)<0$. At $D_j(v)=0$, both decisions admit an optimal completion.
Let $\mathscr I_j := \sigma(X_j^{(0)},G_j^{(0)})$ be the information available to unit $j$'s local assignment. The conditional mean $\ex_\p[D_j(V)\mid\mathscr I_j]$ determines the
optimal decision in a relaxation that restricts only $j$'s assignment to
local information. It does not generally determine a feasible optimal local rule.
The planner must make the same treatment assignment decision across target states with identical local information, even when \(D_j(V)\) changes sign across those target states. Define:
\begin{equation}
    \mathcal B_j(\p)
    :=\ex_\p[(D_j(V))_+]
       -\ex_\p\!\left[
          (\ex_\p[D_j(V)\mid\mathscr I_j])_+
        \right]
\end{equation}
The first term corresponds to the welfare advantage from treating unit $j$ in an optimal global rule. The second term corresponds to the same rule which only uses local information to assign treatment to unit $j$. 
The difference measures the value of the additional flexibility used in global rules when assigning unit \(j\)’s treatment. Theorem~\ref{lem:gap-local-information-bound} bounds the welfare gap through a binary value-of-information argument \citep{blackwell1953equivalent}. 

\begin{theorem}[Bound on the welfare gap]
\label{lem:gap-local-information-bound}
For any distribution $\p$, the bounds hold:
\begin{align}
    \Delta_\p~\geq~\max_{1\leq j\leq\J}\mathcal B_j(\p).
        \label{eq:gap-information-absolute}
\end{align}
The difference $\mathcal B_j(\p)$ is strictly positive if and only if, with positive probability, both
\[
    \p(D_j(V)>0\mid\mathscr I_j)~>~0
    \quad\text{and}\quad
    \p(D_j(V)<0\mid\mathscr I_j)~>~0.
\]
Finally, $\mathcal B_j$ is non-decreasing
under a mean-preserving spread of $D_j(V)$ conditional on $\mathscr I_j$.
\end{theorem}

\begin{runexample}[Welfare gap under additive LSE]
Suppose the target population uses LSE and has mean responses
\begin{equation}\label{eq:gap-additive-lse-response}
    \mu^{(0)}_{(a,s,d)}(\ell)
    =b(\ell)+\tau(\ell)a+\beta(\ell)s,
    \qquad
    \ell=(x,C(d)),\quad (a,s,d)\in C(d),
\end{equation}
where $C(d)=\{0,1\}\times\{0,\ldots,d\}\times\{d\}$ and
$b,\tau,\beta$ are measurable, with responses satisfying the
maintained boundedness restriction. Here $\tau(\ell)$ is the
direct mean treatment effect and $\beta(\ell)$ is the mean
effect of one additional treated neighbor. Define the total
marginal social response to treating unit $j$ by
\begin{equation}\label{eq:gap-social-marginal-response}
    R_j^{\mathrm{soc}}(V)
    :=\tau(L_j^{(0)})
      +\sum_{i\neq j}A_{ij}^{(0)}\beta(L_i^{(0)}).
\end{equation}
For $B<\J$, let $\lambda_{j,B}(V)$ be the $B$th largest element
of $\bigl((R_i^{\mathrm{soc}}(V))_+\bigr)_{i\neq j}$ (including ties).
When $B=\J$, set $\lambda_{j,\J}(V):=0$.

Conditional welfare is an allocation-independent baseline plus
$\J^{-1}\sum_i t_iR_i^{\mathrm{soc}}(V)$. Including $j$ therefore
contributes its social response but displaces the $B$th largest
positive response among the other units, whenever capacity binds.
Thus, defining the budget-adjusted social response,
\begin{equation}
    \Gamma_{j,B}(V)
    :=R_j^{\mathrm{soc}}(V)-\lambda_{j,B}(V), \qquad
    D_j(V)=\frac{\Gamma_{j,B}(V)}{\J}.
      \label{eq:gap-budget-adjusted-response}
\end{equation}
Theorem~\ref{lem:gap-local-information-bound} consequently gives
\begin{equation}\label{eq:gap-social-dispersion-bound}
    \Delta_\p
    \geq\frac1{2\J}
    \max_{1\leq j\leq\J}
    \left\{
        \ex_\p|\Gamma_{j,B}(V)|
        -\ex_\p\!\left|
            \ex_\p[\Gamma_{j,B}(V)\mid\mathscr I_j]
        \right|
    \right\}.
\end{equation}
The bound is strictly positive whenever, for some unit, the sign
of this budget-adjusted social response remains uncertain
conditional on its local information: both signs must have
positive conditional probability on an event of positive probability.

If additionally $|\Gamma_{j,B}(V)|\geq r_j>0$ almost surely
for some unit $j$, then
\begin{equation}\label{eq:gap-social-sign-variance-bound}
    \Delta_\p
    \geq\frac{r_j}{4\J}\,
    \ex_\p\!\left[
        \operatorname{Var}_\p\!\left(
            \operatorname{sgn}\Gamma_{j,B}(V)
            \mid\mathscr I_j
        \right)
    \right].
\end{equation}
Indeed, the two conditional positive-part expectations of
$\Gamma_{j,B}(V)$ and $-\Gamma_{j,B}(V)$ are bounded below
by $r_j$ times their respective sign probabilities; applying
$\min\{p,1-p\}\geq p(1-p)$ yields the displayed inequality.
The sign relevant to the welfare gap concerns direct benefits and spillovers to
other units net of the opportunity cost of capacity, rather
than simply an uncentered mean outcome.$\blacktriangledown$
\end{runexample}

\section{Simulation: Welfare Gap}
\label{sec:simulation_welfare_gap}

The purpose of this simulation exercise is to quantify the welfare loss associated with restricting the policymaker to a local treatment rule when treatment effects exhibit network spillovers. In particular, I compare the population welfare attained by an optimal local rule with that attained by an optimal global rule. Unlike the empirical welfare maximization exercise, the comparison here abstracts from statistical estimation error in policy learning. I first construct an empirically calibrated data-generating process (DGP), treat this DGP as known, and directly optimize population welfare over the local and global policy classes.

Each target population consists of \(J=5\) units. For a capacity level \(B\in\{1,\ldots,5\}\), both policy classes must satisfy the resource constraint
for every feasible realization of the target population. The main object of interest is the relative welfare gap 
\begin{equation}\label{eq:welfare_gap_sim}
 \Delta_\p^{\mathrm{rel}}(B)
       :=\frac{\Delta_\p(B)}{W(\pig^*_{\p,B};\p)},
  \qquad     
\Delta_\p(B):=W(\pi^*_{\p,B};\p)
                -W(\pi_{\lambda^*_{\p,B}};\p),
\end{equation}
where the ratio is reported only when the denominator is positive.
where \(\pi_{\lambda^*_{\p,B}}\) and \(\pig^*_{\p,B}\) denote the welfare-maximizing local and global rules, respectively, subject to the capacity constraint $B$. I compute \(\Delta(B)\) for \(B=1,\ldots,5\).

I construct the distribution of target covariates and networks by resampling from the empirical villages in the pilot data. Let \(\mathcal V_5\) denote the set of observed villages containing at least five units. A target village is generated in two steps. First, I draw a source village uniformly. Conditional on draw \(V=v\), I draw a subset \(S\) of five units uniformly without replacement from the units belonging to village \(v\):
\begin{equation}
    S\mid V=v
    \sim
    \mathrm{Unif}\left(
    \left\{
    s\subseteq v:
    |s|=5
    \right\}
    \right).
\end{equation}
The target covariate matrix and adjacency matrix are then given by
\(X_S, A_{S,S}\) where the latter is the subnetwork induced by the sampled units. This procedure defines an empirical target-state distribution \(P_0(X,A)\) that preserves the joint dependence between observed covariates and network links within villages. Formally, the target-state distribution can be written as
\begin{equation}
    P_0(X,A)
    =
    \sum_{v\in\mathcal V_5}
    \frac{1}{|\mathcal V_5|}
    \frac{1}{\binom{N_v}{5}}
    \sum_{\substack{s\subseteq v\\|s|=5}}
    \mathbf 1
    \left\{
    (X,A)
    =
    (X_s,A_{s,s})
    \right\},
    \label{eq:target_empirical_dgp}
\end{equation}
where \(N_v\) denotes the number of observed units in source village \(v\).

I next construct an empirically calibrated conditional mean outcome function using the pilot data under the LSE exposure specification.
Write \(t\) for own treatment, \(r\) for the number of treated neighbors, and \(\ell\) for degree. I estimate the conditional mean response
\begin{equation}
    \mu(t,r,\ell,x)
    =
    \mathbb E
    \left[
    Y_j
    \mid
    T_j=t,\,
    \sum_k A_{jk}T_k=r,\,
    \sum_k A_{jk}=\ell,\,
    X_j=x
    \right]
\end{equation}
using a logistic regression on the pilot sample. Specifically, the fitted index contains own treatment, the number of treated neighbors, its square, degree, the covariates \(X_j\), and interactions of \(X_j\) with both own treatment and the number of treated neighbors.

Let \(\widehat\mu\) denote the fitted response surface which I subsequently treat
as the true conditional mean response. Thus, the simulation is empirically calibrated, but once the DGP has been constructed, welfare is evaluated without outcome-estimation error.

This exercise isolates the welfare consequences of the policy-class restriction itself. In particular, it removes differences arising from empirical welfare estimation, propensity-score estimation, density-ratio estimation, or finite-sample policy learning. The resulting gap measures the extent to which allowing treatment assignments to depend jointly on the covariates of all units in the target population can improve welfare relative to a pointwise local treatment rule, holding fixed the same resource constraint and the same simulated economic environment. The results are presented below and demonstrate that even in this empirically motivated DGP with limited sign heterogeneity in marginal social response, depending on the value of the constraint, the welfare gap as a fraction of the optimal welfare ranges between 0.3\% (when the constraint is slack) and 18.3\%.

\begin{figure}
    \centering
    \includegraphics[scale=0.2]{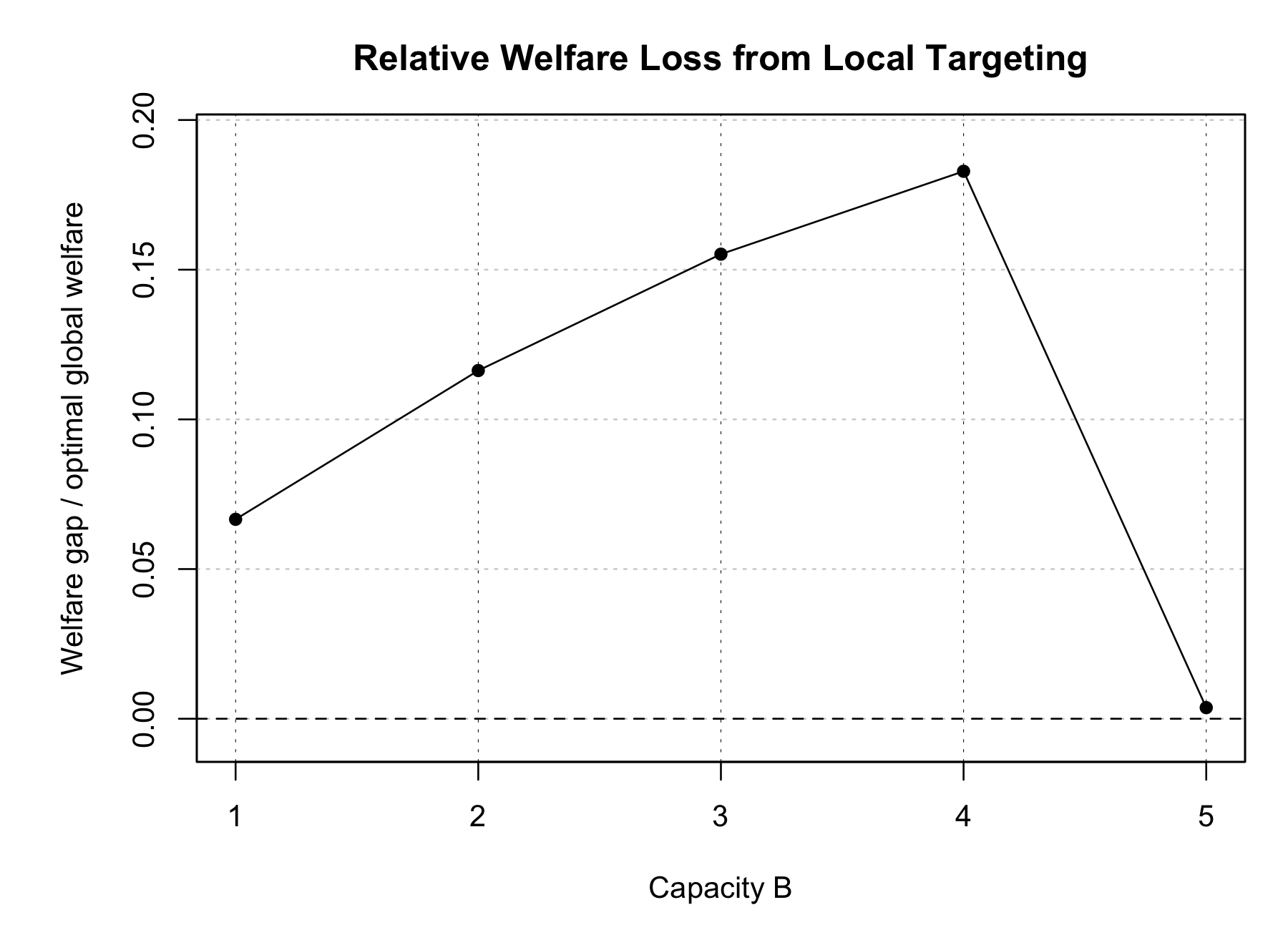}
    \caption{Simulation results}
    \label{fig:placeholder}
\end{figure}

\section{Conclusion}
\label{sec:conclusion}

This paper studies treatment-policy design under network spillovers when a policy is
learned in a pilot population and subsequently deployed in a target population. I
provide conditions under which target-population welfare is identified from pilot data
and develop an empirical welfare maximization procedure that accommodates general
exposure mappings, heterogeneous network structures, and differences between populations.

The regret bounds separate policy-class complexity from restrictions needed to
control dependence in the pilot population. The upper bound requires independent
node-level fundamentals and independent Bernoulli randomization, together with
locality of the complete welfare score. Its network terms account for dyadic
randomness and for changes transmitted through the incident links of a resampled
node. A sublinear exposure diameter alone does not guarantee consistency, and
the upper-bound assumptions are not established here for general epidemic
models with shared transmission shocks. Under sparse local-score conditions and
overlap bounded away from zero, the usual square-root pilot-size rate is recovered
up to logarithmic factors. The lower-bound proof verifies membership in the same common model class.
The rate comparison additionally requires compatible welfare richness and
a nonvanishing treatment fraction. The lower bound does
not establish sharpness of each network-dependence term.


\begin{appendix}

\section{Main Proofs}

\subsection{Proof of Theorem~\ref{lem::iden_welfare}}

\begin{proof}
Fix $\pi\in\Pi$ and pilot unit $i$. By
Assumptions~\ref{asmp2::eff_treat}, \ref{asmp2::pot_out},
\ref{asmp2::consistent}, and~\ref{asmp2::overlap}, together with
conditional unconfoundedness in
Eqn.~\eqref{eqn2::strict_unconfoundedness}, which follows from
Assumptions~\ref{asmp:ub-fundamentals},
\ref{asmp:experimental-design}, and~\ref{asmp2::unconfoundedness}:
\[
    \ex\left[
        \frac{\Y{1}_i\one\{E_i^{(1)}=c\}}
             {e_{i,c}(L_i^{(1)})}
        \,\Bigg|\,L_i^{(1)}=\ell
    \right]
    =\ex[\Y{1}_i(c)\mid L_i^{(1)}=\ell]
    =\mu_c^{(1)}(\ell)
\]
for almost every policy-relevant pair $(\ell,c)$. The last equality
uses Definition~\ref{def:mean-response} and
Assumption~\ref{asmp2::identical}.b. Irrelevant contributions are
defined as zero before division.

Eqn.~\eqref{eq:ub-ipw-weight} and iterated expectations therefore
give
\begin{align*}
    \ex[\Upsilon_i(Z_i;\pi)]
    &=\sum_{c\in\cC^{(0)}}
      \int
      \mu_c^{(1)}(\ell)s(\ell)\overline q(\ell;c,\pi)
      P_L^{(1)}(d\ell) =\sum_{c\in\cC^{(0)}}
      \int
      \mu_c^{(0)}(\ell)\overline q(\ell;c,\pi)
      P_L^{(0)}(d\ell)\\
    =\frac1\J\sum_{j=1}^{\J}
      \ex&\left[
          \sum_{c\in\cC^{(0)}}
          \mu_c^{(0)}(L_j^{(0)})
          \p(E_j^\pi=c\mid L_j^{(0)})
      \right] =\frac1\J\sum_{j=1}^{\J}
      \ex\left[\mu_{E_j^\pi}^{(0)}(L_j^{(0)})\right]
      =\W(\pi;\p).
\end{align*}
The second equality uses stability and site support,
Assumptions~\ref{asmp2::stability} and~\ref{asmp2::support},
through the change-of-measure identity
\citep[Ch.~6, \S~32; Ch.~3, \S~16]{billingsley1995probability}.
The third uses the common local-state marginal distributions in
Assumption~\ref{asmp2::identical}.a and the averaged kernel in
Definition~\ref{defn2::exp_freq_pi}. The final line follows from
iterated expectations, target conditional-mean irrelevance in
Assumption~\ref{asmp2::unconfoundedness}.a, and the welfare definition
in Eqn.~\eqref{eq2:welfare}.

Since this identity holds for every pilot unit, linearity of
expectation yields
\[
    \ex\left[
        \frac1\n\sum_{i=1}^{\n}\Upsilon_i(Z_i;\pi)
    \right]
    =\frac1\n\sum_{i=1}^{\n}
        \ex[\Upsilon_i(Z_i;\pi)]
    =\W(\pi;\p).
\]
Thus empirical welfare is unbiased at any arbitrary fixed rule $\pi$.
\end{proof}

\subsection{Proof of Theorem~\ref{thm:upper-new}}

\begin{proof}
\noindent\emph{Step 1: Policy learning bound.}
The goal is to bound $\sup_{\pi \in \Pi}W(\pi;\p)-W(\widehat\pi;\p)$. Fix $\p\in\Pc$, empirical maximization gives the pathwise comparison:
\begin{equation}\label{eq:upper-step1}
 R(\widehat\pi)
 ~\leq~ 2\sup_{\p\in\Pc}\ex\left[ \sup_{\pi\in\Pi}
                |\widehat W_\n(\pi)-W(\pi;\p)| \right].
\end{equation}
This is the usual empirical-welfare comparison
\citep[see Eqns. 2.2 and 2.3][]{kittet18}; it does not require attainment of the
population supremum.

\noindent\emph{Step 2: Decomposing the sampling error.}
For any fixed DGP $\p$, I decompose the expectation of the supremum of a centered empirical process in the RHS of Eqn. \ref{eq:upper-step1}. Note that the empirical welfare is  determined by fundamentals $(F^{(1)})$ and the realized network $(\bA{1})$.
For $A\in\sA_\n$, evaluate $L_i^{(1)}(F^{(1)},A)$ and
$E_i(F^{(1)},A)$ using the covariates and treatments contained in
$F^{(1)}$ and the maintained feature and exposure transformations. Set
\begin{align}
Z_i(F^{(1)},A)& ~:=~ 
 \bigl(\Y{1}_i(E_i(F^{(1)},A)),L_i^{(1)}(F^{(1)},A),
                                  E_i(F^{(1)},A)\bigr),\nonumber \\
 S_\pi(F^{(1)},A)
    & ~:=~ \tfrac1\n\sum_{i=1}^{\n}\sum_{c\in\cC^{(0)}} \Y{1}_i(c)\,\tfrac{\one\{E_i(F^{(1)},A)=c\}}{e_{i,c}(L_i(F^{(1)},A))}
                     s(L_i(F^{(1)},A))\overline q(L_i(F^{(1)},A),c;\pi), \nonumber\\
 m_\pi(F^{(1)})
   & ~:=~ \ex[S_\pi(F^{(1)},\bA{1})\mid F^{(1)}],\nonumber\\
 S_\pi^0(F^{(1)})& ~:=~ S_\pi(F^{(1)},\mathbf0),\qquad
 r_\pi(F^{(1)}) ~:=~ m_\pi(F^{(1)})-S_\pi^0(F^{(1)}).
 \label{eq:aux-processes}
\end{align}
Here, $m_\pi(F^{(1)})$ is the conditional expectation of the empirical welfare across realizations of the network given fundamentals and $S_\pi^0(F^{(1)})$ is the empirical welfare on an empty network. Finally, $r_\pi(F^{(1)})$ represents the average remainder from the empirical welfare in a randomly realized network compared to the empirical welfare in an empty network given fundamentals. 
Theorem~\ref{lem::iden_welfare} gives $\ex S_\pi(F^{(1)},\bA{1})=W(\pi;\p)$. Further, bounds in Assumptions~\ref{asmp2::pot_out}.b and \ref{asmp2::overlap} imply that for $B_\J=M\kappa/\gamma_\J$: 
\begin{equation*}\label{eq:aux-local-coefficients}
 \Y{1}_i(E_i(F^{(1)},A))\,\tfrac{\one\{E_i(F^{(1)},A)=c\}}{e_{i,c}(L_i(F^{(1)},A))}s(L_i(F^{(1)},A)) ~\leq~ B_\J \quad\text{for all }i, \text{ almost surely}.
\end{equation*}

Iterated expectations give the desired decomposition:
\begin{equation*}\label{eq:ub-exact-decomposition}
\begin{split}
 S_\pi(F^{(1)},\bA{1})~-~&\ex S_\pi(F^{(1)},\bA{1})
 ~=~\underbrace{S_\pi^0(F^{(1)})-\ex S_\pi^0(F^{(1)})}
                           _{\text{independent benchmark}}\\
 &+~\underbrace{S_\pi(F^{(1)},\bA{1})-m_\pi(F^{(1)})}
                           _{\text{dyadic link variation}}\quad+\quad \underbrace{r_\pi(F^{(1)})-\ex r_\pi(F^{(1)})}
                           _{\text{network-mediated node variation}}.
\end{split}
\end{equation*}
The first term corresponds to an empty-network approximation that serves only as an analytical device.
The second term captures the randomness in the realized network relative to the average network. 
The third term captures fluctuations in the graph-averaged correction to the empty-network benchmark through deviations in the network from the empty network. 

\noindent\emph{Step 3: Bound each term.}
The independent benchmark term is the deviation between the empirical welfare in an empty network and its expectation.
In an empty network, spillovers are immaterial and so traditional tools from empirical process theory apply. Lemma~\ref{lem:independent-benchmark} follows symmetrization of the independent vectors $F_i$, not the dependent observations $Z_i$. The Rademacher process
has sub-Gaussian increments in its empirical $L_2$ metric. The entropy integral,
including a zero function to bound an absolute supremum, as proposed in 
\citep{bartlett2002rademacher,vdvwellner1996}. This establishes the square-root overlap dependence using an $L_2$ envelope for
the actual independent benchmark, rather than a pointwise variance assertion
for the dependent process.
\begin{equation*}\label{eq:iid-term-new}
 \ex\sup_{\pi\in\Pi}
 |S_\pi^0(F^{(1)})-\ex S_\pi^0(F^{(1)})|
 \leq CM\kappa\sqrt{\frac{a_\Pi}{\n\gamma_\J}}.
\end{equation*}

Lemma~\ref{lem:dyadic-variation} bounds the concentration of the dyadic link variation term. Conditional on the fundamentals, an edge change can affect at most $d_\n^{\mathrm{sc}}$ score summands,
and each retained summand changes by at most $2B_n$. This term admits a Doob
martingale difference representation across dyads within the network. 
Conditional independence of the dyads therefore gives the predictable
quadratic-variation bound.
Freedman's inequality \citep{freedman1975}, then yields the valid Bernstein-type conditional tail bound for the dyadic variation term. Standard metric entropy bounds then yield the desired result.
\begin{equation*}
 \ex_\p\sup_{\pi\in\Pi}
 |S_\pi(F^{(1)},\bA{1})-m_\pi(F^{(1)})|
 \leq CB_\J d_\n^{\mathrm{sc}}
       \left[\sqrt{\rho_\n\vartheta_\n}
                          +\frac{\vartheta_\n}{\n}\right].
\end{equation*}
Lemma~\ref{lem:node-variation}
bounds the network-mediated node-variation term, which accounts for changes in treatment and covariates transmitted through existing incident links, not merely changes in their probabilities. Since the \(r_\pi(F^{(1)})\) term is generally not an average of functions of individual primitives, an ordinary Rademacher sum argument fails.
The proof instead establishes a bounded-differences property of the entire function \(r_\pi\), then uses a primitive-pair swapping argument.
\begin{equation*}
 \ex_\p\sup_{\pi\in\Pi}
 |r_\pi(F^{(1)})-\ex_\p r_\pi(F^{(1)})|
 ~\leq~ CB_\J d_\n^{\mathrm{sc}}
       \left[\rho_\n\sqrt{\n\vartheta_\n}
                          +\frac{\vartheta_\n}{\n}\right].
\end{equation*}
\noindent\emph{Step 4: Triangle inequality.}
Putting the bounds together, by the triangle inequality,
\begin{equation}\label{eq:regular-uniform-concentration}
\begin{split}
 &\ex_\p\sup_{\pi\in\Pi}|\widehat W_\n(\pi)-W(\pi;\p)|\\
 &\quad\leq CM\kappa\Bigg[
 \sqrt{\frac{a_\Pi}{\n\gamma_\J}}
 +\frac{d_\n^{\mathrm{sc}}}{\gamma_\J}
 \left(\sqrt{\rho_\n\vartheta_\n}
       +\rho_\n\sqrt{{\vartheta_\n}{\n}}
       +\frac{\vartheta_\n}{\n}\right)\Bigg].
\end{split}
\end{equation}
Substitute Eqn.~\eqref{eq:regular-uniform-concentration} into
Eqn.~\eqref{eq:upper-step1}. If $\rho_\n=0$ or
$d_\n^{\mathrm{sc}}=0$, the two network components vanish by the
corresponding lemmas. Finally, target welfare lies in $[-M,M]$, so
regret is at most $2M$. This proves Eqn.~\eqref{eq:upper-new}.
\end{proof}

\begin{proof}[Proof of Corollary~\ref{cor:sparse-bdd-degree}.]
Substitute $\rho_\n\leq C_\rho/\n$,
$d_\n^{\mathrm{sc}}\leq C_d$, $\gamma_\J\geq\gamma_0>0$,
and $\kappa\leq\kappa_0$ into Eqn.~\eqref{eq:upper-new}; use
$a_\Pi\leq\vartheta_\n$. If $\vartheta_\n/\n\leq1$, its linear
term is no larger than its square root. The catalogue condition
$K_{\mathrm{set},\n}\leq\n^r$ gives the logarithmic simplification.
\end{proof}

\subsection{Proof of Theorem~\ref{thm:lower-new}}

\begin{proof}
The lower bound construction is related to the
proof in \citet[Theorem~2.2]{kittet18} except that the testing reduction is an
Assouad type bound \citep{yu1997assouad}. Let $\tilde \pi$ denote any statistical rule, possibly randomized and not necessarily in the class $\Pi$.

\noindent\emph{Step 1: Bound by a model subclass.}
For any finite family of DGPs $\Pc_m \subset \Pc$ that are contained within model class $\Pc$, the following inequality holds:
\[ \sup_{\p\in\Pc}~\ex_\p
      \left[
          \sup_{\pi\in\Pi}\W(\pi;\p)
          -\W(\widetilde\pi;\p)
      \right] ~\geq~ \max_{\p\in\Pc_{m}}~\ex_\p
      \left[
          \sup_{\pi\in\Pi}\W(\pi;\p)
          -\W(\widetilde\pi;\p)
      \right]  \]

I now construct one such finite collection of DGPs $\Pc_m$. Let integer $m$, realizations $x_0,x_1,\ldots,x_m$ and set $N_m:=\{-1,1\}^m$ be as defined in Assumption~\ref{asmp:welfare_shatter}, where $x_0$ is considered inactive while $x_1,\ldots,x_m$ are considered active realizations.
I now characterize a finite family of DGPs $\{ \p_\nu: \nu \in N_m\}$.
Set parameters $0<\delta\leq1/4$, \(
    \alpha:=\frac{B}{2\J},\) and \(
    p_\n:=\min\{2\gamma_\J,1/2\}
\).
Since $1\leq B\leq\J$ and $0<\gamma_\J<1/2$, the following hold:
\begin{equation}\label{eq:lb-assignment-probability}
     0<\alpha\leq1/2, \qquad \gamma_\J<p_\n\leq2\gamma_\J,
    \qquad
    1-p_\n>\gamma_\J.
\end{equation}
For each $\nu \in N_m$, independently draw covariates for units in each population from
\begin{equation}\label{eq:lb-covariate-law}
    F_X(\{x_r\})=\frac{\alpha}{m},\quad r=1,\ldots,m,
    \qquad
    F_X(\{x_0\})=1-\alpha.
\end{equation}
Thus the expected number of active target units is
$\J\alpha=B/2$. Active types occur often enough to affect welfare,
while the budget can accommodate all of them on an event of
non-negligible probability. In the pilot, draw $T_i^{(1)}\sim\mathrm{Bernoulli}(p_\n)$
independently across units and independently of both populations'
other primitives and dyadic uniforms. The upper bound $p_\n\leq2\gamma_\J$
will make informative treated observations rare, whereas the two
strict lower bounds in
Eqn.~\eqref{eq:lb-assignment-probability} ensure overlap.

In both populations, draw $U_i^{(e)}$ and
$\varepsilon_i^{(e)}$ independently (of covariates across units and
populations) from $\mathrm{Unif}[0,1]$.  Setting the common graphon to
$g\equiv0$, Eqn.~\eqref{eq:pilot-graphon} gives
\[
    \bA{e}=\mathbf0_{\nee\times\nee},
    \qquad
    V^{(e)}=v_\varnothing^{(e)}(\bX{e})
    \quad\p_\nu\text{-almost surely}.
\]
This is an admissible submodel of the graphon experiment, not a
restriction of the entire model class to empty networks. Let $P_V$ be the target-state law (which is common across $\nu \in N_m$), and write $\ex_V$ and $\p_V$ for expectation and probability under it.
For each $\nu \in N_m$, mean treatment response functions given $\X{e}_i=x_r$, $r\in\{1,\ldots,m\}$ are defined as follows:
\begin{equation}\label{eq:lb-outcome-law}
    \Y{e}_i(c_{\varnothing,1})
    :=M\left(
        2\one\{\varepsilon_i^{(e)}\leq1/2+\delta\nu_r\}-1
       \right);\quad \Y{e}_i(c) :=0,
    \forall c\in\cC^{(e)}\setminus\{c_{\varnothing,1}\}.
\end{equation}
For $\X{e}_i=x_0$, set every potential outcome equal to zero.
For an active type, the treated outcome mean is therefore
$2M\delta\nu_r$, while the untreated mean is zero. It is easy to verify that this set of DGPs $\{\p_\nu: \nu \in N_m\} \in \Pc$ since they satisfy the assumptions in Definition~\ref{def:uniform-regret}.

For any arbitrary statistical treatment rule $\widetilde\pi$, the followin identity holds:
\begin{equation}\label{eq:lb-finite-family-reduction}
    \sup_{\p\in\Pc}\ex_\p
       \left[
           \sup_{\pi\in\Pi}\W(\pi;\p)
           -\W(\widetilde\pi;\p)
       \right] ~\geq~ 
       \max_{\nu\in N_m}
       \ex_{\p_\nu}[\mathcal R_\nu(\widetilde\pi)]
       ~\geq~
       \frac1{2^m}\sum_{\nu\in N_m}
       \ex_{\p_\nu}[\mathcal R_\nu(\widetilde\pi)].
\end{equation}
where
\(
    \mathcal R_\nu(\widetilde\pi)
    :=\sup_{\pi\in\Pi}\W(\pi;\p_\nu)
       -\W(\widetilde \pi;\p_\nu)
\). 
The last inequality bounds worst-case (within $\Pc_m$) regret with average regret given a uniform prior over $\Pc_m$.

\noindent\emph{Step 2: Decomposition of regret.}
Denote by $w_\nu(v,\mathbf t)$ the conditional welfare in an allocation $\mathbf t\in\Tc$ for a given realization of observed target state $v=v_\varnothing^{(0)}(\mathbf x)$.
\begin{equation}\label{eq:lb-conditional-welfare}
    w_\nu(v,\mathbf t)
    :=\ex_{\p_\nu}\left[
        \frac1\J\sum_{j=1}^{\J}
        \Y{0}_j(E_j^{(0)}(\mathbf t;v))
        \,\Bigg|\,V=v
      \right] 
    =\frac{2M\delta}{\J}
      \sum_{r=1}^{m}\nu_r
      \sum_{j=1}^{\J}\one\{\X{0}_j(v)=x_r\}t_j.
\end{equation}
The second equality follows from Eqn.~\eqref{eq2:welfare} and the constructed means in Eqn.~\ref{eq:lb-outcome-law}. 
An optimal rule will never assign treatments to any unit $j$ such that $\X{0}_j(v)=x_r$ and $\nu_r = -1$. Further, assigning treatments to inactive units with $\X{0}_j(v)=x_0$ is welfare neutral. Therefore, I only need to consider optimal rules which assign treatments to units such that $\X{0}_j(v)=x_r$ and $\nu_r = 1$ (at most $\B$ of them). The rule $\pi_\nu$ in Assumption~\ref{asmp:welfare_shatter} therefore maximizes
conditional welfare over all feasible allocations at every state
in the support. In particular, $\sup_{\pi\in\Pi}\W(\pi;\p_\nu) = \W(\pi_\nu;\p_\nu)$.
Its conditional welfare also dominates that of every feasible policy
outside $\Pi$. 
By law of iterated expectations:
\begin{equation*}
    \mathcal R_\nu(\widetilde\pi) ~=~ \ex_{\p_\nu}\left[ w_\nu(V,\pi_\nu(V)) - w_\nu(V,\widetilde\pi(V)) \right] =: \ex_{\p_\nu}\left[ r_\nu(\pi_\nu,\widetilde\pi; V) \right] 
\end{equation*}
In order to further analyze this quantity in the presence of a capacity constraint, it will be useful to partition the space of target states $\mathcal V_0$. Define:
\[
 N_{\mathrm{act}}(v):=\sum_{j=1}^{\J}\one\{X_j^{(0)}(v)\in\{x_1,\ldots,x_m\}\},
 \qquad \mathcal E_B:=\{v:N_{\mathrm{act}}(v)\leq B\},
\]
where $\mathcal E_B$ as the set of target states in which every policy which does not assign treatments to inactive units is feasible. Thus, for states in $\mathcal E_B$, 
\begin{equation}\label{eqnA9}
    \left[ \pi_\nu(v) \right]_j = \one\left\{ \sum_{r=1}^m \nu_r \one\{\X{0}_j(v)=x_r\} > 0\right\}.
\end{equation}
From Eqns.~\eqref{eq:lb-conditional-welfare} and \eqref{eqnA9}, the regret within global state $v \in \mathcal E_B$ can be written as:
\begin{multline*}
    r_\nu(\pi_\nu,\widetilde\pi; v) =
    \tfrac{2M\delta}{\J}
    \sum_{r=1}^{m}\sum_{j=1}^{\J}
    \one\{\X{0}_j(v)=x_r\}
    \left[
        \one\{\nu_r=1\}(1-\widetilde\pi_j(v))
        +\one\{\nu_r=-1\}\widetilde\pi_j(v)
    \right].
\end{multline*}
The bracket records either failure to treat a positive-effect unit or treatment of a negative-effect unit. Each such mistake costs $2M\delta$ in that units conditional mean outcome. Conditional on $\X{0}_j=x_r$, the number of other active units
is $\mathrm{Binomial}(\J-1,\alpha)$. Thus, by Markov's inequality $b_B:=\p_V(\mathcal E_B\mid\X{0}_j=x_r) \geq\frac12$.
The probability of the relevant conditioning event is
\(
    \p_V(\X{0}_j=x_r,\mathcal E_B)=\tfrac{\alpha b_B}{m}.
\)
Discarding nonnegative regret outside $\mathcal E_B$:
\begin{equation}\label{eqnA10}
    \mathcal R_\nu(\widetilde\pi)
    \geq
    \tfrac{2M\delta\alpha b_B}{m}
    \sum_{r=1}^{m}
    \left[
        \one\{\nu_r=1\}(1-u_r(\widetilde\pi))
        +\one\{\nu_r=-1\}u_r(\widetilde\pi)
    \right]
\end{equation}
where $u_r(\widetilde\pi) :=\frac1\J\sum_{j=1}^{\J}
       \ex_V\left[
           \widetilde\pi_j(V)\mid\X{0}_j=x_r,\mathcal E_B
       \right]$. The quantity $u_r(\widetilde\pi)\in[0,1]$ averages the policy's treatment
probability for type $r$ over unit labels and target states in $\mathcal E_B$.

\noindent\emph{Step 3: Minimax regret bounded by error in testing problem.}
Fix a measurable feasible policy $\widetilde\pi$, not necessarily in $\Pi$.
For each active type, define a decoder:
\begin{equation}\label{eq:lb-policy-decoder}
    \mathscr D_r(\widetilde\pi)
    :=
    \begin{cases}
        1,&u_r(\widetilde\pi)\geq1/2,\\
        -1,&u_r(\widetilde\pi)<1/2.
    \end{cases}
\end{equation}
The decoder depends only on $\widetilde\pi$ and the common law $P_V$, not
on $\nu$. Its conditioning events have positive probability,
and the target support is finite. Therefore the decoder of a jointly
measurable statistical rule is measurable in the observed sample
and any independent auxiliary randomization. 
Applying Eqn.~\eqref{eq:lb-policy-decoder}, a decoding error bound is attained:
\begin{equation}
    \mathcal R_\nu(\widetilde\pi)
    \geq
    \tfrac{M\delta\alpha}{2m}
    \sum_{r=1}^{m}
    \one\{\mathscr D_r(\widetilde\pi)\neq\nu_r\}.
    \label{eq:lb-welfare-decoding}
\end{equation}
Thus a low-regret policy necessarily yields few sign errors.

\noindent\emph{Step 4: Testing error bound.}
Let $\mathbb Q_\nu$ be the law of the complete observed pilot
data $\mathbf Z$ and the statistical rule's independent auxiliary
randomization $\zeta$. The known design, transport weights, and
policy-indexed target kernels are common across alternatives and
add no information about $\nu$.
Under the construction, the triples
$(\X{1}_i,T_i^{(1)},\Y{1}_i)$ are i.i.d. The conditional KL identity and additivity
across the triples presented in Lemma~\ref{lem:lb-divergence} gives
\begin{equation}
    KL(\mathbb Q_\nu,\mathbb Q_{\nu^{(r)}})
     \leq\tfrac{32\n\alpha\gamma_\J\delta^2}{m}.
    \label{eq:lb-adjacent-KL}
\end{equation}
The inequality uses the constructed upper bound $p_\n\leq2\gamma_\J$ on treatment probability.

\noindent\emph{Step 5: Choose the signal size for the test.}
Set $\delta =\tfrac18
    \left[
        1\wedge\sqrt{\tfrac{m}{\n\alpha\gamma_\J}}
    \right]$.
Then $\delta\leq1/8<1/4$, and
Eqn.~\eqref{eq:lb-adjacent-KL} implies
\begin{equation*}
    KL(\mathbb Q_\nu,\mathbb Q_{\nu^{(r)}})
    \leq
      \tfrac{32\n\alpha\gamma_\J}{m}\tfrac1{64}
      \min\left\{1,\tfrac{m}{\n\alpha\gamma_\J}\right\} =\tfrac12
      \min\left\{\tfrac{\n\alpha\gamma_\J}{m},1\right\}
      \leq\tfrac12.
\end{equation*}
Larger $\delta$ increases the welfare cost of a mistake but makes
the alternatives easier to distinguish. The chosen value preserves
a nontrivial welfare separation while keeping the information about
each sign bounded. The expected number of informative pilot
observations for one sign is $\n\alpha p_\n/m$. 
Lemma~\ref{lem:assouad-testing} gives the Assouad testing bound
\begin{equation}\label{eq:lb-testing-errors}
    \tfrac1{2^m}\sum_{\nu\in N_m}
    \sum_{r=1}^{m}
    \mathbb Q_\nu\{\widetilde\nu_r\neq\nu_r\}
    \geq\tfrac m4.
\end{equation}

\medskip
\noindent\emph{Step 6: Convert testing errors into minimax regret.}
Fix any measurable feasible statistical treatment rule
$\widetilde\pi$, and define
$\widetilde\nu_r:=\mathscr D_r(\widetilde\pi)$.
For each realization of the pilot data and auxiliary randomization,
Eqn.~\eqref{eq:lb-welfare-decoding} applies to the output policy.
Taking expectations, then applying
Eqns.~\eqref{eq:lb-finite-family-reduction}
and~\eqref{eq:lb-testing-errors}, gives
\begin{equation*}
    \sup_{\p\in\Pc}\ex_\p
      \left[
          \sup_{\pi\in\Pi}\W(\pi;\p)
          -\W(\widetilde\pi;\p)
      \right]  \geq
      \tfrac{M\delta\alpha}{2m}
      \tfrac1{2^m}\sum_{\nu\in N_m}
      \sum_{r=1}^{m}
      \mathbb Q_\nu\{\widetilde\nu_r\neq\nu_r\} \geq
      \tfrac{M\delta\alpha}{2m}\frac m4
      =\tfrac{M\delta\alpha}{8}.
\end{equation*}
Substituting $\delta =\tfrac18
    \left[
        1\wedge\sqrt{\tfrac{m}{\n\alpha\gamma_\J}}
    \right]$ and
$\alpha=B/(2\J)$ yields
\( \frac{M\delta\alpha}{8}
    \geq\frac M{128}
      \left[
          \frac B\J\wedge
          \sqrt{\frac{Bm}{\J\n\gamma_\J}}
      \right].\)
The inequality holds for arbitrary $\widetilde\pi$, so taking
the infimum proves \eqref{eq:lower-new}.
\end{proof}

\subsection{Proof of Theorem~\ref{lem:gap-local-information-bound}}
\begin{proof}
\noindent\emph{Step 1: Characterizing the welfare from the optimal global rule.}
Since $\Tc$ is finite, and the set of global rules $\overline{\Pi}$ is unrestricted, the welfare optimal global rule picks conditional welfare maximizing allocations statewise. That is, for any fixed unit $j$, the optimal global rule will allocate treatment to unit $j$ in target state $v$, if and only if $D_j(v)$ is positive. Then, this conditional welfare maximizing allocation satisfies:
\[
    \max_{\mathbf t\in\Tc}w_\p(V,\mathbf t)
    = w^*_{j,0}(V) + (D_j(V))_+.
\]
Thus, the welfare from the optimal global rule satisfies 
\begin{equation}\label{eqnglobalopt}
    \W(\pig^*_\p;\p)
=
\mathbb E_{\mathbb P}[w^*_{j,0}(V)]
+
\mathbb E_{\mathbb P}[(D_j(V))_+].
\end{equation}
\noindent\emph{Step 2: Bounding the welfare from the optimal local rule.}
Now, I introduce an intermediate allocation problem.
For a fixed unit \(j\), consider all feasible global rules which satisfy $\mathscr I_j$-measurability of the $j^{th}$ component, $[\pi(V)]_j$. 
Every other coordinate may still depend on the full observed state \(V\). There is no symmetry requirement in this intermediate allocation problem.
Let its optimal welfare be \(W^{\mathrm{rel},*}_j(\mathbb P)\). The policy classes are nested as follows:
\[ \Lambda_B(\p) ~\subseteq~ \{ \pi \in \overline{\Pi}: [\pi(V)]_j \text{ is }\mathscr I_j\text{-measurable}\} ~\subseteq~ \overline{\Pi}.\] 
Therefore, the optimal welfare within these classes are also ordered:
\[
\W\bigl(\pig_{\pil^*_\p};\p\bigr)
~\leq~
W^{\mathrm{rel},*}_j\bigl(\p\bigr)
~\leq~
\W\bigl(\pig^*_\p;\p\bigr) \quad \imp \quad \Delta_{\mathbb P}
~\geq~
\W\bigl(\pig^*_\p;\p\bigr) -W^{\mathrm{rel},*}_j\bigl(\p\bigr).\]

In the intermediate problem, let \(a_j\in\{0,1\}\) be an \(\mathscr I_j\)-measurable inclusion decision while every other assignment is optimized using \(V\). Its welfare is 
\(
\mathbb E_{\mathbb P}
[w^*_{j,0}(V)+a_jD_j(V)].
\)

By iterated expectations,
\[
\mathbb E_{\mathbb P}[a_jD_j(V)]
=
\mathbb E_{\mathbb P}
\left[
a_j\mathbb E_{\mathbb P}[D_j(V)\mid\mathscr I_j]
\right].
\]
The optimal inclusion decision is therefore
\(
a_j^*
=
\mathbf1\left\{
\mathbb E_{\mathbb P}[D_j(V)\mid\mathscr I_j]>0
\right\}
\).
Consequently, the optimal welfare in the intermediate problem is:
\begin{equation}\label{eqnrelative}
W^{\mathrm{rel},*}_j(\mathbb P)
=
\mathbb E_{\mathbb P}[w^*_{j,0}(V)]
+
\mathbb E_{\mathbb P}
\left[
\left(
\mathbb E_{\mathbb P}[D_j(V)\mid\mathscr I_j]
\right)_+
\right].
\end{equation}
Eqns.~\eqref{eqnglobalopt} and \eqref{eqnrelative} hold for every target unit $j$ and so too for the maximum.  Subtracting the two gives the desired bound in Theorem~\ref{lem:gap-local-information-bound}.

\noindent\emph{Step 3: Positivity of $\mathcal B_j(\p)$.}
An informative equivalent expression for $\mathcal B_j(\p)$ follows from \(
D_j=(D_j)_+-(-D_j)_+, \ 
|D_j|=(D_j)_++(-D_j)_+
\), then:
\begin{align*}
    \mathcal B_j(\p) &= 
    \ex_\p\left[ \ex_\p\left[ (D_j(V))_+\mid\mathscr I_j \right] - \left( \ex_\p\left[ (D_j(V))_+\mid\mathscr I_j \right] - \ex_\p\left[ (-D_j(V))_+\mid\mathscr I_j \right] \right)_+ \right]
    \\
    &= \ex_{\p} \min\left\{ \ex_{\p}[(D_j(V))_+\mid\mathscr I_j], \ex_{\p}[(-D_j(V))_+\mid\mathscr I_j]  \right\}.
\end{align*}
The above is strictly positive exactly where local rules leaves both a positive and a negative inclusion advantage possible. Thus
\(
\mathcal B_j(\mathbb P)>0
\)
if and only if 
\(
\p(D_j(V)>0\mid\mathscr I_j)>0
\ \text{and}\ 
\p(D_j(V)<0\mid\mathscr I_j)>0
\) hold.
That is, it is not enough that treatment advantages have different signs somewhere in the population. Their signs must remain unresolved after conditioning on the local information that local rules use.
\end{proof}

\bibliographystyle{ecta-fullname} 
\bibliography{reference}  

\newpage
\section{Auxiliary lemmas}

\subsection{Auxiliary lemmas for the upper bound}
\label{app:upper-auxiliary}

\begin{auxlemma}[Entropy of policy-induced exposure probabilities]
\label{lem:kernel-entropy}
Suppose Assumption~\ref{asmp:Nat} holds, and write
$\overline q_\pi(\ell,c):=\overline q(\ell,c;\pi)$.
For every probability law $R$ on
$\mathcal L\times\cC^{(0)}$,
\begin{equation}
\label{eq:entropy-new}
\forall \varepsilon \in (0,1), \quad 
 \log N\!\left(
   \varepsilon,
   \{\overline q_\pi:\pi\in\Pi\},
   L_2(R)
 \right)
 \leq
 2N_{dim}\log\!\left(\frac{6e\Q^2}{\varepsilon^2}\right).
\end{equation}
\end{auxlemma}

\begin{proof}
Fix the DGP $\p$. Identical allocations at a target state
generate identical exposures. Hence, for every $\ell,c,\pi,\pi'$,
\[
 \left|
   \overline q_\pi(\ell,c)-\overline q_{\pi'}(\ell,c)
 \right|
 \leq
 \frac1\J\sum_{j=1}^{\J}
 \p\!\left(
   \pi(V)\neq\pi'(V)
   \,\middle|\,
   L_j^{(0)}=\ell
 \right).
\]
Since the right-hand side belongs to $[0,1]$, squaring and
integrating gives
\begin{equation*}
 \|\overline q_\pi-\overline q_{\pi'}\|_{L_2(R)}^2
 \leq
 \frac1\J\sum_{j=1}^{\J}
 \int
 \p\!\left(
   \pi(V)\neq\pi'(V)
   \,\middle|\,
   L_j^{(0)}=\ell
 \right)R(d\ell,dc) =: d_{\p,R}(\pi,\pi').
\end{equation*}
This is the disagreement pseudometric under the probability
distribution obtained by drawing $(\ell,c)$ from $R$, selecting
$j$ uniformly from $[\J]$, and drawing a target state from its
conditional distribution under $\p$ given $L_j^{(0)}=\ell$.
The multiclass Sauer bound \citep{hausslerlong1995} and the
corresponding packing argument imply
\[
 \log N(\delta,\Pi,d_{\p,R})
 \leq
 2N_{dim}\log\!\left(\frac{6e\Q^2}{\delta}\right),
 \qquad 0<\delta\leq1.
\]
An internal $\varepsilon^2$-cover in $d_{\p,R}$ induces an
internal $\varepsilon$-cover of
$\{\overline q_\pi:\pi\in\Pi\}$ in $L_2(R)$.
Taking $\delta=\varepsilon^2$ proves the result.
\end{proof}

\begin{auxlemma}[Independent benchmark bound]
\label{lem:independent-benchmark}
Under Assumptions in Theorem~\ref{thm:upper-new}, there is a universal constant $C$ such that:
\begin{align}\label{eq:aux-weighted-benchmark}
 \ex_\p\sup_{\pi\in\Pi}
 |S_\pi^0(F^{(1)})-\ex_\p S_\pi^0(F^{(1)})|
 &\leq CM\kappa \sqrt{\frac{a_\Pi}{\n \gamma_\J}}.\\
 \text{ where } \quad w_i^0(F_i^{(1)}) &:= \sum_{t \in \{0,1\}}\Y{1}_i(c_{\varnothing,t})\,\tfrac{\one\{T_i^{(1)} = t\}}{e_{i,c_{\varnothing,t}}(\ell_\varnothing(\X{1}_i))} s(\ell_\varnothing(\X{1}_i)) \nonumber
\end{align}
\end{auxlemma}

\begin{proof}
Under Assumption~\ref{asmp:ub-locality}, in any empty network, the local-state--exposure pair simplifies to self covariate and the two feasible exposures in any empty network that correspond to the two possible self treatment states $L_i(F^{(1)},\mathbf0) = \ell_\varnothing(\X{1}_i) := (\X{1}_i,C_\varnothing)$. Consequently, the empirical welfare admits the additive separability over fundamentals of each separate unit as $S_\pi^0(F^{(1)}) = \n^{-1}\sum_{i=1}^{\n}b_{i,\pi}(F_i^{(1)})$, where
\begin{equation*}
    b_{i,\pi}(F_i^{(1)})
   := \sum_{t \in \{0,1\}} \Y{1}_i(c_{\varnothing,t})\,\tfrac{\one\{T_i^{(1)} = t\}}{e_{i,c_{\varnothing,t}}(\ell_\varnothing(\X{1}_i))}
                     s(\ell_\varnothing(\X{1}_i)) \overline q\bigl(\ell_\varnothing(\X{1}_i),c_{\varnothing,t};\pi\bigr).
\end{equation*}
An advantage of this representation is that fundamentals $F_i^{(1)}$ are independent across units under Assumptions~\ref{asmp:ub-fundamentals} and \ref{asmp:experimental-design}. For an independent Rademacher process $\epsilon_i$, the symmetrization inequality holds \citep[Lemma 2.3.1]{vdvwellner1996}:
\begin{equation*}
    \ex\sup_{\pi \in \Pi} \left| \tfrac{1}{\n} \sum_{i=1}^\n \left[ b_{i,\pi}(F_i^{(1)}) - \ex b_{i,\pi}(F_i^{(1)}) \right] \right| ~ \leq ~ 2 \ex \sup_{\pi \in \Pi} \left| \tfrac{1}{\n}\sum_{i=1}^\n \epsilon_i  b_{i,\pi}(F_i^{(1)}) \right|
\end{equation*}
For fixed realization of fundamentals $F^{(1)}$, I denote the conditional Rademacher process:
\begin{equation*}
    H_\pi ~:=~ \tfrac{1}{\n} \sum_{i=1}^\n \epsilon_i  b_{i,\pi}(F_i^{(1)}).
\end{equation*}
It is easy to verify that the process is sub-Gaussian, that is, for every real $s > 0$:
\begin{align*}
    \ex_{\epsilon|F^{(1)}} \exp\{s(H_\pi - H_{\pi'}) \} ~&\leq~ \exp\left[ \frac{t^2d_{F^{(1)}}(\pi, \pi')^2}{2\n} \right]\\
    \text{ where } 
d_F(\pi,\pi')^2
&:=
\frac1{n_1}\sum_i(b_{i,\pi}-b_{i,\pi'})^2.
\end{align*}
has sub-Gaussian increments in \(d_F/\sqrt{\n}\) metric. Applying a maximal inequality bound \citep[Lemma 2.2.5]{vdvwellner1996} yields:
\begin{multline*}
    \ex_{\epsilon|F^{(1)}} \sup_{\pi \in \Pi} \left| \tfrac{1}{\n}\sum_{i=1}^\n \epsilon_i  b_{i,\pi}(F_i^{(1)}) \right| ~\leq~ \\
    \frac{C_0}{\sqrt{\n}}\left\{\frac1\n\sum_{i=1}^{\n}
                    (w_i^0(F_i^{(1)}))^2\right\}^{1/2}\left[ 1 + \int_0^1 \sqrt{ \log N\!\left(\varepsilon,
       \{\overline q_\pi:\pi\in\Pi\},L_2(R)\right)} d\varepsilon \right]
\end{multline*}
for any probability law $R$ on
$\mathcal L\times\cC^{(0)}$ and $0<\varepsilon\leq1$.
Lemma~\ref{lem:kernel-entropy} guarantees:
\begin{equation}\label{eq:ub-disagreement-cover}
 \log N\!\left(\varepsilon,
       \{\overline q_\pi:\pi\in\Pi\},L_2(R)\right)
 \leq 2N_{dim}\log(6e\Q^2/\varepsilon^2).
\end{equation}
Here an internal cover has its centers in $\{\overline q_\pi:\pi\in\Pi\}$, and $e$ is Euler's
number. The result follows from the multiclass Sauer bound of
\citet{hausslerlong1995} and a packing argument.
The entropy integral, with the zero function added to
bound an absolute supremum, yields
\citep{bartlett2002rademacher,vdvwellner1996}
\[
 \ex_{\epsilon|F^{(1)}} \sup_{\pi\in\Pi}
 \left|\frac1\n\sum_{i=1}^{\n}\epsilon_i b_{i,\pi}(F_i^{(1)})\right|
 \leq C_0\sqrt{\frac{a_\Pi}{\n}}
    \left\{\frac1\n\sum_{i=1}^{\n}
                    (w_i^0(F_i^{(1)}))^2\right\}^{1/2}.
\]
Law of iterated expectations, Jensen's inequality and the bound $\ex[w_i^0(F_i^{(1)})]^2 \leq \frac{2M^2\kappa^2}{\gamma_\J}$ prove
Eqn.~\eqref{eq:aux-weighted-benchmark}. 
\end{proof}

\begin{auxlemma}[Dyadic variation bound]
\label{lem:dyadic-variation}
Under Assumptions in Theorem~\ref{thm:upper-new}, there is a universal constant $C$ such that:
\begin{equation}\label{eq:ub-edge-bound}
 \ex_\p\sup_{\pi\in\Pi}
 |S_\pi(F^{(1)},\bA{1})-m_\pi(F^{(1)})|
 \leq CB_\J d_\n^{\mathrm{sc}}
       \left[\sqrt{\rho_\n\vartheta_\n}
                          +\frac{\vartheta_\n}{\n}\right].
\end{equation}
\end{auxlemma}

\begin{proof}
Condition on a realization of $F^{(1)}$. Suppose $A$ and $A'$ are adjacency matrices corresponding to networks that differ on exactly one dyadic link. By Assumption~\ref{asmp:ub-locality}.b, this changes the exposure of at most $d_\n^{\mathrm{sc}}$ nodes. Therefore,
\begin{equation}\label{eq:ub-edge-increment}
 |S_\pi(F^{(1)},A)-S_\pi(F^{(1)},A')|
 \leq \frac{2B_\J d_\n^{\mathrm{sc}}}{\n}.
\end{equation}
Enumerate all $\binom\n2$ undirected dyads and index the corresponding dyads by $k$. 
Denote the corresponding entries of adjacency matrix $A$ by $A_{k}$. Denote by $\{\mathcal H_k : 0 \leq k \leq \binom\n2\}$ a filtration of the sigma-algebra generated by $(F^{(1)},\bA{1})$ where $\mathcal H_k := \sigma(F^{(1)},A_{1},\ldots,A_{k})$ and $\mathcal H_0 := \sigma(F^{(1)})$. I use this to represent $S_\pi(F^{(1)},\bA{1})$ as a sum of martingale differences. 
\begin{equation*}
    S_\pi(F^{(1)},\bA{1})-m_\pi(F^{(1)}) = \sum_{k=0}^{\binom\n2}D_{k,\pi}; \ \ D_{k,\pi} = M_{k,\pi} - M_{k-1,\pi}; \ \  M_{k+1,\pi} = \ex[S_\pi(F^{(1)},\bA{1})|\mathcal H_k]
\end{equation*}
The Doob martingale increment has the form
$D_{k,\pi}=(A_{k}-p_k)v_{k,\pi}$, where $p_k = \rho_\n g\!\left(
            \X{1}_i,\X{1}_k,U_i^{(1)},U_k^{(1)}\right) $ is the probability of link formation and from \ref{eq:ub-edge-increment} $|v_{k,\pi}|\leq \frac{2B_\J d_\n^{\mathrm{sc}}}{\n}$.
Conditional dyadic independence and $p_k\leq\rho_\n$ give quadratic variation bound:
\begin{equation}\label{eq:ub-edge-variance}
 \sum_{k=1}^{\binom\n2}\ex_\p[D_{k,\pi}^2\mid\mathcal H_{k-1}]
 \leq\binom\n2\rho_\n \frac{2B_\J d_\n^{\mathrm{sc}}}{\n}^2
 \leq 2B_\J^2(d_\n^{\mathrm{sc}})^2\rho_\n
 =:\sigma_\n^2.
\end{equation}
Freedman's inequality \citep{freedman1975} implies
\[
 \p\!\left(
 |S_\pi(F^{(1)},\bA{1})-m_\pi(F^{(1)})|
 >\sqrt{2\sigma_\n^2t}+\frac{2a_\n t}{3}
 \,\middle|\,F^{(1)}\right)\leq2e^{-t},\qquad t>0.
\]
To translate this into a bound for the supremum, I approximate it with a finite maximum. 
Every local state generated by
varying the adjacency matrix belongs to $\mathscr L_{F^{(1)}}
    :=\left\{
        (X_i^{(1)},C):
        i\in[\n],\ C\in\mathscr C_\n
      \right\}$, with $|\mathscr L_{F^{(1)}}|
    \leq\n K_{\mathrm{set},\n}$.
Apply the finite-state approximation in Lemma~\ref{lem:kernel-entropy} at radius $1/\n$ to obtain a finite policy net, 
$\Pi_{F^{(1)}}\subseteq\Pi$ such that every $\pi\in\Pi$ has a representative $\pi'\in\Pi_{F^{(1)}}$ satisfying $ \max_{\ell\in\mathscr L_{F^{(1)}}}\max_{c\in\cC^{(0)}}
 |\overline q(\ell;c,\pi)-\overline q(\ell;c,\pi')|
 \leq\varepsilon$. 
Moreover, by Eqn.~\eqref{eq:ub-complexity},
\begin{equation}\label{eq:ub-dyadic-net-cardinality}
    \log|\Pi_{F^{(1)}}|
    \leq
    2N_{dim}\log\!\left(
        6e\Q^2\n^2K_{\mathrm{set},\n}
    \right)
    \leq2\vartheta_\n.
\end{equation}
The net depends only on the conditioned primitives, not on the realized dyads.
Define $N_{F^{(1)}}:=\log\!\left(2|\Pi_{F^{(1)}}|\right)$,  then $\log2\leq N_{F^{(1)}}\leq3\vartheta_\n$.
Applying the preceding conditional Freedman inequality to each representative and using a union
bound, for every $u\geq0$:
\begin{equation}\label{eq:ub-dyadic-net-tail}
    \p\Bigg(
       \max_{\pi\in\Pi_{F^{(1)}}}
        \left|
            S_\pi(F^{(1)},\bA{1})-m_\pi(F^{(1)})
        \right|
       >
        \sqrt{2\sigma_\n^2(N_{F^{(1)}}+u)}
        +\frac{2a_\n}{3}(N_{F^{(1)}}+u)
        \,\Bigg|\,F^{(1)}
    \Bigg)
    \leq e^{-u}.
\end{equation}

Integrating ~\eqref{eq:ub-dyadic-net-tail} gives the bound:
\begin{align*}
    &\ex\left[
        \max_{\pi\in\Pi_{F^{(1)}}}
        \left|
            S_\pi(F^{(1)},\bA{1})-m_\pi(F^{(1)})
        \right|
        \,\Bigg|\,F^{(1)}
      \right]\\
    &\quad\leq
      \sqrt{2\sigma_\n^2N_{F^{(1)}}}
      +\frac{2a_\n N_{F^{(1)}}}{3}\ +\  
      \int_0^\infty e^{-u}
      \left\{
          \sqrt{\frac{\sigma_\n^2}
                       {2(N_{F^{(1)}}+u)}}
          +\frac{2a_\n}{3}
      \right\}\,du\\
    &\quad\leq
      \sqrt{2\sigma_\n^2N_{F^{(1)}}}
      +\sqrt{\frac{\sigma_\n^2}{2N_{F^{(1)}}}}
      +\frac{2a_\n}{3}(N_{F^{(1)}}+1)\ \leq\ 
      C\left\{
          \sqrt{\sigma_\n^2\vartheta_\n}
          +a_\n\vartheta_\n
      \right\}.
\end{align*}
The first inequality follows from the tail-integral formula, splitting at the threshold corresponding to $u=0$ and then changing variables to the threshold in
Eqn.~\eqref{eq:ub-dyadic-net-tail}. 
The last inequality uses $\log2\leq N_{F^{(1)}}\leq3\vartheta_\n$. Finally, approximate the supremum over $\Pi$:
\begin{equation}\label{eq:ub-dyadic-net-reduction}
    \sup_{\pi\in\Pi}
      \left|
          S_\pi(F^{(1)},\bA{1})-m_\pi(F^{(1)})
      \right|~\leq~
      \max_{\pi'\in\Pi_{F^{(1)}}}
      \left|
          S_{\pi'}(F^{(1)},\bA{1})-m_{\pi'}(F^{(1)})
      \right|
      +\frac{2B_\J}{\n}.
\end{equation}

Combining this bound with Eqn.~\eqref{eq:ub-dyadic-net-reduction}, taking expectations over $F^{(1)}$, and substituting $a_\n$ and $\sigma_\n^2$ from Eqns.~\eqref{eq:ub-edge-increment} and~\eqref{eq:ub-edge-variance}, respectively, yields
\begin{align}
    \ex\sup_{\pi\in\Pi}
      \left|
          S_\pi(F^{(1)},\bA{1})-m_\pi(F^{(1)})
      \right|\ &\leq\ 
      C\left\{
          \sqrt{\sigma_\n^2\vartheta_\n}
          +a_\n\vartheta_\n
      \right\}
      +\frac{2B_\J}{\n} \\
      &\leq \ 
      CB_\J d_\n^{\mathrm{sc}}
      \left\{
          \sqrt{\rho_\n\vartheta_\n}
          +\frac{\vartheta_\n}{\n}
      \right\}.
      \label{eq:ub-edge-bound}
\end{align}
The approximation remainder is absorbed using
$d_\n^{\mathrm{sc}}\geq1$ and $\vartheta_\n\geq1$.
The argument may first be applied to finite subsets of the
countable approximating class in Assumption~\ref{asmp:Nat};
the bounds are uniform in their cardinality, so the same
approximation argument used for the entropy bound extends the
result to $\Pi$. If $d_\n^{\mathrm{sc}}=0$ or $\rho_\n=0$, the centered dyadic
process is identically zero, and neither the net nor its
approximation remainder is needed.
\end{proof}

\begin{auxlemma}[Network-mediated node variation bound]
\label{lem:node-variation}
Under Assumptions in Theorem~\ref{thm:upper-new}, there is a universal constant $C$ such that:
\begin{equation}\label{eq:ub-node-bound}
 \ex_\p\sup_{\pi\in\Pi}
 |r_\pi(F^{(1)})-\ex_\p r_\pi(F^{(1)})|
 ~\leq~ CB_\J d_\n^{\mathrm{sc}}
       \left[\rho_\n\sqrt{\n\vartheta_\n}
                          +\frac{\vartheta_\n}{\n}\right].
\end{equation}
\end{auxlemma}
\begin{proof}
Fix $k\in[\n]$, and let $F^{(1),[k]}$ differ from $F^{(1)}$
only in node $k$'s fundamentals. Couple graph draws $A$ and $A^{[k]}$
under the two primitive arrays using the same dyadic uniforms, $\xi$. All dyads not involving $k$ then agree. Deleting dyads involving node $k$ from both draws yields
the same matrix $A^{-k}$.
By Assumption~\ref{asmp:ub-locality}.a,
\begin{equation}\label{eq:ub-isolation-cancellation}
    S_\pi(F^{(1)},A^{-k})
       -S_\pi(F^{(1)},\mathbf0)\ = \ 
     S_\pi(F^{(1),[k]},A^{-k})
       -S_\pi(F^{(1),[k]},\mathbf0).
\end{equation}
Since node $k$ is isolated in $A^{-k}$,  its contribution cancels within
each difference. For $i\neq k$, replacing $F_k^{(1)}$ does not
change the local-state--exposure pair or potential outcomes for $i$.
Let $\ex_\xi$ denote expectation over the dyadic uniforms,
holding fundamentals fixed. Combining
~\eqref{eq:aux-processes} and~\eqref{eq:ub-isolation-cancellation}, one can decompose $r_\pi(F^{(1)})-r_\pi(F^{(1),[k]})$ into
\begin{align*}
    \ex_\xi\Big[
        S_\pi(F^{(1)},A)-S_\pi(F^{(1)},A^{-k}) 
        -S_\pi(F^{(1),[k]},A^{[k]})
        +S_\pi(F^{(1),[k]},A^{-k})
    \Big].
\end{align*}
Restoring the incident edges one at a time and applying
Eqn.~\eqref{eq:ub-edge-increment} gives
\[
    \left|
        S_\pi(F^{(1)},A)-S_\pi(F^{(1)},A^{-k})
    \right|
    \leq \frac{2B_\J d_\n^{\mathrm{sc}}}{\n} \operatorname{deg}_k(A), \quad \text{where }\operatorname{deg}_k(A):=\sum_{i\neq k}A_{ki}
\]
and the analogous inequality holds for $A^{[k]}$.
Under Assumption~\ref{asmp:ub-fundamentals},
\[
    \ex_\xi[\operatorname{deg}_k(A)]
    \leq(\n-1)\rho_\n,
    \qquad
    \ex_\xi[\operatorname{deg}_k(A^{[k]})]
    \leq(\n-1)\rho_\n.
\]
Combining gives a bound on the edge-averaged remainder, conditional on fundamentals:
\begin{equation}\label{eq:ub-residual-node-increment}
\begin{split}
    \left|
        r_\pi(F^{(1)})-r_\pi(F^{(1),[k]})
    \right| ~\leq ~ \frac{
        4B_\J d_\n^{\mathrm{sc}}(\n-1)\rho_\n
       }{\n}
      =:\lambda_\n.
\end{split}
\end{equation}
It also accounts for treatment changes transmitted
through existing links, even when the coupled graphs coincide.
Let $F^{(1)\prime}$ be an independent copy of $F^{(1)}$.
By Jensen's inequality:
\begin{equation}\label{eq:ub-node-ghost-comparison}
    \ex\sup_{\pi\in\Pi}
      \left|
          r_\pi(F^{(1)})-\ex[r_\pi(F^{(1)})]
      \right|\ \leq \ 
      \ex\sup_{\pi\in\Pi}
      \left|
          r_\pi(F^{(1)})-r_\pi(F^{(1)\prime})
      \right|.
\end{equation}
Under Assumptions~\ref{asmp:ub-fundamentals}
and~\ref{asmp:experimental-design}, the pairs
$(F_i^{(1)},F_i^{(1)\prime})$ are exchangeable and independent.
Let $\epsilon_1,\ldots,\epsilon_\n$ be independent Rademacher
variables. Define
\[
    F_i^{(1),\epsilon}
    :=
    \begin{cases}
        F_i^{(1)},&\epsilon_i=1,\\
        F_i^{(1)\prime},&\epsilon_i=-1,
    \end{cases}
    \qquad
    F^{(1),\epsilon}
    :=(F_1^{(1),\epsilon},\ldots,F_\n^{(1),\epsilon}),
\]
and let $F^{(1),-\epsilon}$ select the complementary member of
each pair. Swapping invariance implies
\begin{equation}\label{eq:ub-swap-symmetrization}
\begin{split}
    &\ex\sup_{\pi\in\Pi}
      \left|
          r_\pi(F^{(1)})-r_\pi(F^{(1)\prime})
      \right|\ = \ 
      \ex_{F^{(1)},F^{(1)\prime}}\ex_\epsilon
      \sup_{\pi\in\Pi}
      \left|
          r_\pi(F^{(1),\epsilon})
          -r_\pi(F^{(1),-\epsilon})
      \right|.
\end{split}
\end{equation}

Condition on $(F^{(1)},F^{(1)\prime})$. Every local state obtained
from a hybrid primitive array and an arbitrary graph belongs to
\[
    \mathscr L_{F^{(1)},F^{(1)\prime}}
    :=
    \left\{
        (X_i^{(1)},C),(X_i^{(1)\prime},C):
        i\in[\n],\ C\in\mathscr C_\n
    \right\}.
\]
Its cardinality is at most $2\n K_{\mathrm{set},\n}$.
Lemma~\ref{lem:kernel-entropy}, applied at radius
$1/\n$, therefore supplies a finite net
$\Pi_{F^{(1)},F^{(1)\prime}}\subseteq\Pi$, with cardinality,
\begin{equation}\label{eq:ub-node-net-cardinality}
    \log|\Pi_{F^{(1)},F^{(1)\prime}}|
    \leq
    2N_{dim}\log\!\left(
        12e\Q^2\n^2K_{\mathrm{set},\n}
    \right)
    \leq2\vartheta_\n.
\end{equation}

For each representative, define
\(
    T_\pi(\epsilon)
    :=
    r_\pi(F^{(1),\epsilon})
    -r_\pi(F^{(1),-\epsilon}).
\)

Since $T_\pi(-\epsilon)=-T_\pi(\epsilon)$ and the signs have
a symmetric distribution,
\(
    \ex_\epsilon[T_\pi(\epsilon)]=0.
\)

If $\epsilon^{[k]}$ differs from $\epsilon$ only in its
$k$th sign, then one node's fundamental changes in each
hybrid array. Eqn.~\eqref{eq:ub-residual-node-increment}
therefore implies
\begin{equation}\label{eq:ub-node-sign-sensitivity}
    |T_\pi(\epsilon)-T_\pi(\epsilon^{[k]})|
    \leq2\lambda_\n.
\end{equation}
The exponential-moment form of the McDiarmid inequality
\citep{mcdiarmid1989} yields
\begin{equation}\label{eq:ub-node-sign-mgf}
    \log\ex_\epsilon
        \exp\{tT_\pi(\epsilon)\}
    \leq
    \frac{t^2\n\lambda_\n^2}{2},
    \qquad t\in\R.
\end{equation}

For any $t>0$, Eqn.~\eqref{eq:ub-node-sign-mgf} gives
\begin{align*}
    \ex_\epsilon
      \max_{\pi\in\Pi_{F^{(1)},F^{(1)\prime}}}
      |T_\pi(\epsilon)|\ \leq\ 
      \tfrac1t\log
      \sum_{\pi\in\Pi_{F^{(1)},F^{(1)\prime}}}
      \left\{
          \ex_\epsilon e^{tT_\pi(\epsilon)}
          +\ex_\epsilon e^{-tT_\pi(\epsilon)}
      \right\}\ \leq \ 
      \tfrac{
          \log(2|\Pi_{F^{(1)},F^{(1)\prime}}|)
      }{t}
      +\tfrac{t\n\lambda_\n^2}{2}.
\end{align*}
Optimizing over $t$ and using
Eqn.~\eqref{eq:ub-node-net-cardinality},
\begin{equation}\label{eq:ub-node-net-maximum}
    \ex_\epsilon
      \max_{\pi\in\Pi_{F^{(1)},F^{(1)\prime}}}
      |T_\pi(\epsilon)|\ \leq \ 
      \lambda_\n
      \sqrt{2\n\log(2|\Pi_{F^{(1)},F^{(1)\prime}}|)}
      \ \leq\  C\lambda_\n\sqrt{\n\vartheta_\n}.
\end{equation}
Since $r_\pi=m_\pi-S_\pi^0$, approximating each residual costs
at most $2B_\J/\n$. Comparing the two complementary hybrid
arrays therefore gives
\[
    |T_\pi(\epsilon)-T_{\pi'}(\epsilon)|
    \leq\frac{4B_\J}{\n}.
\]
Combining this inequality with
Eqns.~\eqref{eq:ub-node-ghost-comparison},
\eqref{eq:ub-swap-symmetrization}, and
\eqref{eq:ub-node-net-maximum}, and then substituting
Eqn.~\eqref{eq:ub-residual-node-increment}, yields
\begin{align}
    &\ex\sup_{\pi\in\Pi}
      \left|
          r_\pi(F^{(1)})-\ex[r_\pi(F^{(1)})]
      \right|\ \leq \ 
      C\lambda_\n\sqrt{\n\vartheta_\n}
      +\frac{4B_\J}{\n}\nonumber\\
    &\quad\leq
      CB_\J d_\n^{\mathrm{sc}}
      \left\{
          \rho_\n\sqrt{\n\vartheta_\n}
          +\frac{\vartheta_\n}{\n}
      \right\},
      \qquad d_\n^{\mathrm{sc}}\geq1.
      \label{eq:ub-node-bound}
\end{align}
The approximation remainder is absorbed using
$d_\n^{\mathrm{sc}}\geq1$ and $\vartheta_\n\geq1$.
As in the entropy argument, the proof can first be applied to
finite subsets of the countable approximating class in
Assumption~\ref{asmp:Nat}. The bounds are uniform in their
cardinality, and the same pointwise approximation and dominated
convergence argument extends the result to $\Pi$.

\end{proof}

\label{app:upper-auxiliary-proofs}

\subsection{Auxiliary lower-bound lemmas}
\label{app:lower-auxiliary}

\begin{auxlemma}[Adjacent-alternative divergence]
\label{lem:lb-divergence}
Let $\mathbb Q_{\nu,\delta}$ be the law of the complete observed
pilot data and independent auxiliary randomization in the binary
subexperiment, and let $\nu^{(r)}$ reverse the $r$th sign of $\nu$.
Under Assumptions in Theorem~\ref{thm:lower-new}, for $0<\delta\leq1/4$,
\begin{equation*}
 KL(\mathbb Q_{\nu,\delta},\mathbb Q_{\nu^{(r)},\delta})
 =\frac{\n\alpha p_\n}{m}
      \operatorname{kl}(1/2+\delta,1/2-\delta) \leq\frac{16\n\alpha p_\n\delta^2}{m}.
\end{equation*}
Here $KL$ denotes Kullback--Leibler divergence, and
$\operatorname{kl}$ is its Bernoulli specialization, using natural
logarithms.
\end{auxlemma}

\begin{proof}
Known information is common to all alternatives and adds
no divergence. Although observed network features may depend on the full
covariate vector, they are a common deterministic function of that vector
in this submodel. The pilot data contain all the triples
$(\X{1}_i,T_i^{(1)},\Y{1}_i)$, which are i.i.d. under the construction.
Hence the full-data divergence is exactly the divergence of these triples;
independent auxiliary randomization leaves it unchanged.

Only
observations with $\X{1}_i=x_r$ and $T_i^{(1)}=1$ have a distribution
that differs between adjacent alternatives. This event has probability
$\alpha p_\n/m$. Then
\[
 KL(\mathbb Q_{\nu,\delta},\mathbb Q_{\nu^{(r)},\delta})
 =\frac{\n\alpha p_\n}{m}
       \operatorname{kl}(1/2+\delta,1/2-\delta).
\]
The same expression applies for either sign of $\nu_r$, since the two
Bernoulli probabilities are complementary. For $\delta\leq1/4$,
\[
    \operatorname{kl}(1/2+\delta,1/2-\delta)
    =2\delta\log\!\left(\frac{1+2\delta}{1-2\delta}\right)
    \leq\frac{8\delta^2}{1-2\delta}
    \leq16\delta^2.
\]
Substitution proves Eqn.~\eqref{eq:lb-adjacent-KL}.
Only the treated-exposure distribution is perturbed; control outcomes
contain no information about the labels.

\end{proof}

\begin{auxlemma}[Assouad testing reduction]
\label{lem:assouad-testing}
Let $\{\mathbb Q_\nu:\nu\in N_m\}$ be a family of laws
on a common observation space. Under Assumptions in Theorem~\ref{thm:lower-new}, if
\(
 KL(\mathbb Q_\nu,\mathbb Q_{\nu^{(r)}})\leq\frac12,\text{ for all }\nu,r,
\)
then every measurable estimator $\widetilde\nu$ of the sign vector
satisfies
\begin{equation*}
 \frac1{2^m}\sum_{\nu\in N_m}
       \sum_{r=1}^{m}\mathbb Q_\nu
                  \{\widetilde\nu_r\neq\nu_r\}
 \geq\frac m4.
\end{equation*}
\end{auxlemma}

\begin{proof}
For total variation defined by $\|P-Q\|_{TV} := \sup_A|P(A)-Q(A)|$, Pinsker's inequality \citep[Lemma~2.5(i), p.~88]{tsybakov2009introduction} gives
$\|\mathbb Q_\nu-\mathbb Q_{\nu^{(r)}}\|_{TV}\leq1/2$.
The two-point testing inequality then implies
\[
 \mathbb Q_\nu\{\widetilde\nu_r\neq\nu_r\}
 +\mathbb Q_{\nu^{(r)}}
                \{\widetilde\nu_r\neq\nu^{(r)}_r\}
 \geq1-\|\mathbb Q_\nu-\mathbb Q_{\nu^{(r)}}\|_{TV}
 \geq\frac12.
\]
For a fixed coordinate, sum over its $2^{m-1}$ disjoint pairs and divide
by $2^m$. This gives average error at least $1/4$ for that coordinate.
Summing over $r=1,\ldots,m$ proves Eqn.~\eqref{eq:lb-testing-errors},
the coordinate-wise Assouad bound \citep{yu1997assouad} see also
\citet[Lemma~2.12 and Theorem~2.12(ii), pp.~117--118]
{tsybakov2009introduction}. 
\end{proof}

\section{Empirical Welfare with a Local-State Cutoff}
\label{app:cutoff-welfare}

A local-state cutoff trades improved exposure overlap against omitted
welfare. The target welfare $\W$, the policy class $\Pi$, and the feasible
allocation set $\Tc$ remain unchanged. The analysis uses the same maintained
model as the main text; its positive finite-sample overlap bound may
converge to zero along a sequence of population sizes.

\subsection{Retained welfare and identification}

\begin{definition}[Empirical welfare with a local-state cutoff]
\label{def:cutoff-welfare}
Let $h:\mathcal L\to\{0,1\}$ be fixed before observing the pilot sample, and invariant across policies. Define
\begin{align}
 \Upsilon_i^h(Z_i;\pi)&:=h(L_i^{(1)})\Upsilon_i(Z_i;\pi),
    \label{eq:cutoff-score}\\
 \widehat W_{\n,h}(\pi)&:=\frac1\n\sum_{i=1}^{\n}\Upsilon_i^h(Z_i;\pi),
    \label{eq:cutoff-empirical}\\
 W_h(\pi;\p)&:=\ex_\p\left[\frac1\J\sum_{j=1}^{\J}
                h(L_j^{(0)})Y_j^{(0)}(E_j^\pi)\right],
    \label{eq:cutoff-population}\\
 \beta_h(\p)&:=\int(1-h(\ell))P_L^{(0)}(d\ell).
    \label{eq:cutoff-target-mass}
\end{align}
An omitted contribution is set to zero before division. An empirical
maximizer is a measurable selection
\begin{equation}\label{eq:cutoff-EWM}
 \widehat\pi_h\in\argmax_{\pi\in\Pi}\widehat W_{\n,h}(\pi),
\end{equation}
which is assumed to exist and be attained.
\end{definition}

The cutoff removes contributions to welfare, rather than units or links. Exposures are
computed on the original network, and normalization remains by $\n$, not
by the number retained. A rule may treat an omitted unit to benefit its
retained neighbors. Unlike trimming that redefines the target population
\citep{crump2009limited}, the regret below is evaluated under the same welfare as before $\W$.

\begin{auxlemma}[Identification and approximation of retained welfare]
\label{lem:retained-identification}
Under the conditions of Theorem~\ref{lem::iden_welfare}, for every fixed
$h$, policy $\pi$, and pilot unit $i$,
\begin{align}
 \ex_\p[\Upsilon_i^h(Z_i;\pi)]&=W_h(\pi;\p),
    \label{eq:cutoff-identification}\\
 \sup_{\pi\in\Pi}|W(\pi;\p)-W_h(\pi;\p)|&\leq M\beta_h(\p).
    \label{eq:cutoff-bias}
\end{align}
\end{auxlemma}
\begin{proof}
The identification argument in Theorem~\ref{lem::iden_welfare}, with
$h(\ell)$ included in the integrand, gives
\[
 \ex_\p\Upsilon_i^h
 =\sum_{c\in\cC^{(0)}}\int h(\ell)\mu_c^{(1)}(\ell)
       s(\ell)\overline q(\ell,c;\pi)P_L^{(1)}(d\ell)
 =W_h(\pi;\p).
\]
Only pairs in $\mathcal E_\pi(\p)$ with $h(\ell)=1$ enter the identity.
For every policy, bounded outcomes give
\[
 |W(\pi;\p)-W_h(\pi;\p)|
 \leq M\ex_\p\left[\frac1\J\sum_j(1-h(L_j^{(0)}))\right]
 =M\beta_h(\p).
\]
\end{proof}

\subsection{Regret under the original welfare criterion}

For a nonempty subclass $\Pc_\star\subseteq\Pc$, write
\[
 R_\star(\widetilde\pi):=
 \sup_{\p\in\Pc_\star}\ex_\p\left[
       \sup_{\pi\in\Pi}W(\pi;\p)-W(\widetilde\pi;\p)\right].
\]

\begin{theorem}[Upper bound for cutoff empirical welfare]
\label{thm:cutoff-upper}
Suppose the conditions of Theorem~\ref{thm:upper-new} hold
uniformly on $\Pc_\star$, and fix a retention function $h$ common to this
subclass. Suppose $\gamma_{\J,h}>0$ satisfies Assumption~\ref{asmp2::overlap}, for every $\p\in\Pc_\star$,
$i$, $c$, and $\pi\in\Pi$,
\begin{equation}\label{eq:cutoff-overlap}
 e_{i,c}(\ell)\geq\gamma_{\J,h}
 \quad\text{for }P_L^{(1)}\text{-almost every }\ell
 \text{ with }(\ell,c)\in\mathcal E_\pi(\p),\ h(\ell)=1.
\end{equation}
Let $\overline\beta_h:=\sup_{\p\in\Pc_\star}\beta_h(\p)$. There is a
universal constant $C$ such that every measurable maximizer in
\eqref{eq:cutoff-EWM} satisfies
\begin{equation}\label{eq:cutoff-upper}
 R_\star(\widehat\pi_h)
 \leq 2M\wedge\Bigg\{2M\overline\beta_h+CM\kappa\Bigg[
 \sqrt{\frac{a_\Pi}{\n\gamma_{\J,h}}}+\frac{d_\n^{\mathrm{sc}}}{\gamma_{\J,h}}
 \left(\sqrt{\rho_\n\vartheta_\n}
       +\rho_\n\sqrt{\n\vartheta_\n}
       +\frac{\vartheta_\n}{\n}\right)\Bigg]\Bigg\}.
\end{equation}
The network term is omitted if $\rho_\n=0$ or $d_\n^{\mathrm{sc}}=0$.
\end{theorem}

\begin{proof}
For each realized sample and each comparison policy, add and subtract
$W_h$ and $\widehat W_{\n,h}$. Exact empirical maximization and the preceding
lemma imply
\begin{equation}\label{eq:cutoff-policy-learning}
 R_\star(\widehat\pi_h)
 \leq 2M\overline\beta_h+
 2\sup_{\p\in\Pc_\star}\ex_\p\sup_{\pi\in\Pi}
       |\widehat W_{\n,h}(\pi)-W_h(\pi;\p)|.
\end{equation}
For $\ell=(x,C)$, use the policy-independent coefficients
\[
 a_i^h(\ell,c):=
 \begin{cases}
 h(\ell)s(\ell)/e_{i,c}(\ell),&
 c\in C\cap\cC^{(0)},\ s(\ell)>0,
 e_{i,c}(\ell)\geq\gamma_{\J,h},\\
 0,&\text{otherwise}.
 \end{cases}
\]
Analogues arguments to Theorem~\ref{thm:upper-new} imply that for $B_\J=M\kappa/\gamma_{\J,h}$: 
\begin{equation*}\label{eq:aux-local-coefficients}
 \Y{1}_i(E_i(F^{(1)},A))a_i^h(L_i(F^{(1)},A),E_i(F^{(1)},A)) ~\leq~ B_\J \quad\text{for all }i, \text{ almost surely}.
\end{equation*}

Apply the decomposition and the three component arguments in
Lemmas~\ref{lem:independent-benchmark}, \ref{lem:dyadic-variation}, and
\ref{lem:node-variation}, with these coefficients and
$\gamma_{\J,h}$ in place of $\gamma_\J$. Identification centers the
process at $W_h$. The resulting concentration bound is the bracket in
\eqref{eq:cutoff-upper}, times $CM\kappa$. A dyad can change a retention
indicator only by changing its local state, so the same full-pair locality
bound and isolation cancellation apply. Substitute into
\eqref{eq:cutoff-policy-learning} and use the welfare cap $2M$.
\end{proof}

The case $h\equiv1$ recovers the regular bound. Define
\begin{equation}\label{eq:cutoff-stochastic-scale}
 r_{\n,h}:=\kappa\left[
 \sqrt{\frac{a_\Pi}{\n\gamma_{\J,h}}}
 +\frac{d_\n^{\mathrm{sc}}}{\gamma_{\J,h}}
 \left(\sqrt{\rho_\n\vartheta_\n}
       +\rho_\n\sqrt{\n\vartheta_\n}
       +\frac{\vartheta_\n}{\n}\right)\right],
\end{equation}
with the same zero-network convention. For deterministic $h_\n$,
$r_{\n,h_\n}\to0$ and $\overline\beta_{h_\n}\to0$ imply consistency.
The stronger condition $\overline\beta_{h_\n}=o(r_{\n,h_\n})$ makes
omitted welfare negligible relative to this stochastic upper bound.

\subsection{Degree cutoffs under LSE}

Suppose both populations use LSE. Degree is encoded in the local state
$\ell=(x,C(d))$, where $C(d)=\{0,1\}\times\{0,\ldots,d\}\times\{d\}$.
For deterministic $0\leq k\leq\min\{\n-1,\J-1\}$, set
\begin{equation}\label{eq:degree-cutoff-definition}
 h_k(\ell):=\one\{\operatorname{deg}(\ell)\leq k\},
 \qquad \widehat\pi_k:=\widehat\pi_{h_k},
\end{equation}
with zero extension outside LSE local states. For fixed
$0<p_0\leq1/2$, let
\begin{equation}\label{eq:cutoff-balanced-class}
 \Pc(p_0):=\{\p\in\Pc:p_0\leq p_\n\leq1-p_0\},
 \qquad R_{p_0}:=R_\star\ \text{with }\Pc_\star=\Pc(p_0).
\end{equation}
Define
\begin{equation}\label{eq:degree-tail-bound}
 b_{\n,k}:=
 \begin{cases}
 0,&k\geq\J-1,\\
 1\wedge\kappa\Pr\{\operatorname{Binomial}(\n-1,\rho_\n)>k\},&k<\J-1.
 \end{cases}
\end{equation}
The second branch is also zero at $k=\n-1$.

\begin{cor}[Sparse LSE with a degree cutoff]
\label{cor:degree-cutoff}
Under the preceding conditions,
\begin{equation}\label{eq:degree-cutoff-inputs}
 \gamma_{\J,h_k}=p_0^{k+1},
 \qquad \sup_{\p\in\Pc(p_0)}\beta_{h_k}(\p)\leq b_{\n,k}.
\end{equation}
If $(\n-1)\rho_\n\leq\bar\lambda<\infty$, then for $\n\geq2$,
\begin{equation}\label{eq:degree-cutoff-upper}
 R_{p_0}(\widehat\pi_k)
 \leq2M\wedge\left\{2Mb_{\n,k}+CM\kappa\left[
 \sqrt{\frac{a_\Pi}{\n p_0^{k+1}}}
 +p_0^{-(k+1)}\left(
       \sqrt{\frac{\vartheta_\n}{\n}}+\frac{\vartheta_\n}{\n}\right)
 \right]\right\},
\end{equation}
where $C$ depends only on $\bar\lambda$. If $k+1>\bar\lambda>0$,
\begin{equation}\label{eq:degree-chernoff}
 b_{\n,k}\leq1\wedge\kappa
       \left(\frac{e\bar\lambda}{k+1}\right)^{k+1}.
\end{equation}
If $\bar\lambda=0$, omitted target mass is zero for every $k\geq0$.
\end{cor}
\begin{proof}
For retained $\ell=(x,C(d))$ and $c=(a,b,d)$,
\[
 e_{i,c}(\ell)=p_\n^a(1-p_\n)^{1-a}
          \binom dbp_\n^b(1-p_\n)^{d-b}
 \geq p_0^{d+1}\geq p_0^{k+1}.
\]
On the original probability space, pilot degree is bounded pathwise by
$S_i:=\sum_{j\neq i}\one\{\xi_{\min\{i,j\},\max\{i,j\}}^{(1)}\leq\rho_\n\}$,
which has the $\operatorname{Binomial}(\n-1,\rho_\n)$ law. Hence
\[
 \beta_{h_k}(\p)
 =\int\one\{\operatorname{deg}(\ell)>k\}s(\ell)P_L^{(1)}(d\ell)
 \leq\kappa\Pr(S_i>k).
\]
Target degree is at most $\J-1$, and the mass is at most one.
Exponential Markov inequality, using
$\ex e^{tS_i}\leq\exp\{\bar\lambda(e^t-1)\}$ and
$t=\log((k+1)/\bar\lambda)$, gives \eqref{eq:degree-chernoff}.
If $\bar\lambda=0$, support transfers the pilot's zero-degree property
to the target local-state law. Finally, LSE has $d_\n^{\mathrm{sc}}=2$
for $\n\geq2$, and $\n\rho_\n\leq2\bar\lambda$.
Substitution into Theorem~\ref{thm:cutoff-upper} proves the regret bound.
\end{proof}

\section{A Two-Period SIR Model with Independent Fundamentals}
\label{app:sir}

This section constructs a two-period susceptible--infectious--recovered (SIR) model satisfying Assumptions~\ref{asmp2::eff_treat}--\ref{asmp:ub-fundamentals}.
The model permits secondary transmission: a unit infected in the first period can infect another unit in the second. 

\subsection{Node primitives and potential contacts}

Fix a population $e\in\{0,1\}$ of size $n_e$, the population superscript is suppressed until it is needed. 
The evaluation dates are $s=0,1,2$, with two transmission periods, $0\to1$ and $1\to2$. 
For each unit $i$, let $I_{i,0}\in\{0,1\}$ indicate baseline infection and let $\omega_i\in(0,1)$ be an observed contact-priority characteristic.
Set
\[
 X_i=(I_{i,0},\omega_i),\qquad
 S_{i,0}=1-I_{i,0},\qquad R_{i,0}=0.
\]
The primitive vectors $(U_i,X_i,\varepsilon_i)$ are IID across units.
For a concrete specification, take $I_{i,0}\sim\mathrm{Bernoulli}(\alpha)$, $0<\alpha<1$, and let $U_i,\omega_i,\varepsilon_i$ be mutually independent $\mathrm{Unif}(0,1)$ variables, independent of $I_{i,0}$.
Baseline infections are treated as exogenous seeds, rather than outcomes of an earlier epidemic on the same network. 
The shock $\varepsilon_i$ is drawn once and held fixed across all counterfactual allocations.

Independently of node primitives, draw mutually independent $\xi_{ik}\sim\mathrm{Unif}(0,1)$ for $i<k$. 
For a common measurable symmetric graphon $g:\Xc^2\times[0,1]^2\to[0,1]$ and $0\leq\rho_{n_e}<1$:
\begin{equation}
\label{eq:app_sir2_graphon}
 A_{ik}
 :=\one\{\xi_{ik}\leq\rho_{n_e}g(X_i,X_k,U_i,U_k)\},
 \qquad A_{ki}=A_{ik},\qquad A_{ii}=0.
\end{equation}
Matrix $A$ records static potential contacts network. 
No deterministic maximum-degree restriction is imposed on $A$. The two populations may be generated independently using this construction and their respective sizes and sparsity parameters.

\subsection{A treatment-invariant active-contact rule}

Let $\mathcal N_i(A):=\{k\neq i:A_{ik}=1\}$ denote the neighborhood set. 
Order its members by $(\omega_k,k)$ lexicographically, so that unit indices resolve priority ties. 
Let $\mathcal N_i^{[2]}(A,\mathbf X)$ contain the first two members, or all members if $|\mathcal N_i(A)|<2$. Define the active-contact matrix by mutual selection:
\begin{equation}
\label{eq:app_sir2_active}
 \widetilde A_{ik}
 :=\one\{k\in\mathcal N_i^{[2]}(A,\mathbf X)\}
   \one\{i\in\mathcal N_k^{[2]}(A,\mathbf X)\},
 \qquad \widetilde A_{ii}=0.
\end{equation}
Thus $\widetilde A$ is symmetric, a subgraph of $A$, and satisfies
\begin{equation}
\label{eq:app_sir2_degree}
 \max_i\sum_{k\neq i}\widetilde A_{ik}\leq2.
\end{equation}
Only active contacts transmit infection. The selection rule is fixed
before treatment, uses only observed information, and is recomputed when
$A$ or the covariates are perturbed. It is part of the maintained exposure
transformation, not a conditioning event imposed on the graphon draw.
In particular, $\widetilde A$ need not have conditionally independent
edges: Equation~\eqref{eq:app_sir2_graphon}, and hence
Assumption~\ref{asmp:ub-fundamentals}, applies to $A$.

A sufficient observed feature is 
\(
 G_i=(i,A,\mathbf X),
\)
using its finite-dimensional vector encoding. This is a measurable
transformation of the observed network and covariates and allows the
planner to compute every exposure below.

\subsection{Two-period epidemic dynamics}

Treatment $\mathbf t\in\{0,1\}^{n_e}$ is assigned once at baseline and
remains fixed. It completely prevents acquisition of infection and
suppresses transmission by an initially infectious treated unit.
An infectious unit has one opportunity to transmit and then recovers;
newly infected units become infectious at the next evaluation date.
There is no reinfection or external infection after date zero.

During $0\to1$, an untreated susceptible unit is infected with certainty
if it has an untreated, initially infectious active neighbor. Write
\begin{equation}
\label{eq:app_sir2_first}
 H_i(\mathbf t)
 :=(1-I_{i,0})(1-t_i)
   \one\!\left\{
      \sum_{k\neq i}\widetilde A_{ik}I_{k,0}(1-t_k)>0
   \right\}.
\end{equation}
The date-one states are
\begin{equation}
\label{eq:app_sir2_date1}
 I_{i,1}(\mathbf t)=H_i(\mathbf t),\qquad
 S_{i,1}(\mathbf t)=(1-I_{i,0})(1-H_i(\mathbf t)),\qquad
 R_{i,1}(\mathbf t)=I_{i,0}.
\end{equation}
In particular, $H_i(\mathbf t)$ is a known function of the baseline
configuration and the candidate allocation; it contains no unobserved
infection shock.

During $1\to2$, the infectious sources are precisely those units infected
in the first period. Define their number among $i$'s active neighbors by
\begin{equation}
\label{eq:app_sir2_sources}
 K_i(\mathbf t)
 :=\sum_{k\neq i}\widetilde A_{ik}H_k(\mathbf t)
 \in\{0,1,2\}.
\end{equation}
Fix a nonrandom transmission parameter $\beta\in(0,1)$, common to the two
populations. Define
\begin{align}
 I_{i,2}(\mathbf t)
 &:=S_{i,1}(\mathbf t)
   \one\!\left\{
      \varepsilon_i\leq
      (1-t_i)\bigl[1-(1-\beta)^{K_i(\mathbf t)}\bigr]
   \right\},
   \label{eq:app_sir2_second}\\
 S_{i,2}(\mathbf t)
 &:=S_{i,1}(\mathbf t)-I_{i,2}(\mathbf t),\qquad
 R_{i,2}(\mathbf t):=I_{i,0}+H_i(\mathbf t).
 \label{eq:app_sir2_date2}
\end{align}
Conditional on the baseline configuration and the allocation, an untreated
susceptible unit with $m$ infectious active neighbors is infected in the
second period with probability $1-(1-\beta)^m$.
Equations~\eqref{eq:app_sir2_date1}--\eqref{eq:app_sir2_date2} produce
binary states satisfying $S_{i,s}+I_{i,s}+R_{i,s}=1$ at each date.
The parameter $\beta$ may be unknown to the planner; it is not required to
compute the exposure mapping.

\subsection{Exposure mapping and bounded potential outcomes}

Use the common finite exposure space
\begin{equation}
\label{eq:app_sir2_space}
 \cC^{(0)}=\cC^{(1)}=\cC_{\mathrm{SIR}}
 :=\{0,1\}\times\{0,1\}\times\{0,1,2\}.
\end{equation}
The exposure mapping is
\begin{equation}
\label{eq:app_sir2_exposure}
 E_i(\mathbf t;v)
 :=\bigl(t_i,H_i(\mathbf t;v),K_i(\mathbf t;v)\bigr).
\end{equation}
Its coordinates record own treatment, first-period infection, and the
number of infectious active neighbors entering the second period.
Although the last two coordinates describe intermediate epidemic states,
they are deterministic functions of observed baseline information and
$\mathbf t$. Thus computing $E_i(\mathbf t;v)$ requires neither
post-treatment observations nor the shocks $\varepsilon_k$.
Following Definition~\ref{def2::feas_exp}, set
\begin{equation}
\label{eq:app_sir2_feasible}
 C_i(v)
 :=\{E_i(\mathbf t;v):\mathbf t\in\{0,1\}^{n_e}\},
 \qquad L_i(v):=(X_i(v),C_i(v)).
\end{equation}
The feasible set is nonempty and finite. I evaluate health by the fraction of dates at which the unit is not
infectious:
\begin{equation}
\label{eq:app_sir2_health}
 Y_i(\mathbf t)
 :=1-\frac{I_{i,0}+I_{i,1}(\mathbf t)+I_{i,2}(\mathbf t)}{3}.
\end{equation}
For every $c=(a,h,m)\in\cC_{\mathrm{SIR}}$, define
\begin{equation}
\label{eq:app_sir2_response}
\begin{split}
 Y_i(a,h,m):=1-\frac13\Bigl[ I_{i,0}
 +(1-I_{i,0})\Bigl\{
 h+(1-h)\one\!\left\{
 \varepsilon_i\leq(1-a)\bigl[1-(1-\beta)^m\bigr]
 \right\}\Bigr\}\Bigr].
\end{split}
\end{equation}
This also defines outcomes at labels infeasible for a particular unit.

\subsection{Verification of Assumptions 1--4}

\begin{enumerate}
    \item \emph{Known exposure mapping.}
The sorting, contact selection, and finite sums in
Equations~\eqref{eq:app_sir2_active}, \eqref{eq:app_sir2_first},
and~\eqref{eq:app_sir2_sources} are known measurable functions of
$(\mathbf t,G_i)$. Equation~\eqref{eq:app_sir2_exposure} therefore
satisfies Assumption~\ref{asmp2::eff_treat}.
    \item \emph{Exposure-indexed and bounded responses.}
Since $I_{i,0}H_i(\mathbf t)=0$, substitution into
Equation~\eqref{eq:app_sir2_response} gives, simultaneously over allocations,
\begin{equation}
\label{eq:app_sir2_effective}
 Y_i(\mathbf t)=Y_i\bigl(E_i(\mathbf t)\bigr).
\end{equation}
Thus allocations inducing the same exposure produce the same outcome.
For every exposure label, including infeasible labels,
\[
 \frac23\leq Y_i(c)\leq1.
\]
Consequently, both parts of Assumption~\ref{asmp2::pot_out} hold with
$M=1$.
\item \emph{Locality.}
If $i$ is isolated in $A$, it is also isolated in $\widetilde A$, so
$H_i(\mathbf t)=K_i(\mathbf t)=0$ for every allocation. Hence
\[
 E_i(\mathbf t)=c_{\varnothing,t_i},\qquad
 c_{\varnothing,0}:=(0,0,0),\quad
 c_{\varnothing,1}:=(1,0,0),\quad
 C_i=C_\varnothing:=\{c_{\varnothing,0},c_{\varnothing,1}\}.
\]
These labels are distinct and common across units and populations.
Moreover, if $k$ is isolated in $A$, changing $(X_k,t_k)$ while holding
$A$ and all other inputs fixed cannot change another unit's nomination
list or active contacts: $k$ is not a potential neighbor of any unit.
It also cannot change another unit's first-period infection or
second-period source count. Therefore, for every $i\neq k$, the pair
$(L_i,E_i)$ is unchanged. An isolated unit's own exposure is likewise
unaffected by other units' covariates and treatments. 

Hold the covariate array fixed and let $A,A'$ differ only in dyad
$\{u,v\}$. Only the nomination lists of $u$ and $v$ can change.
At each endpoint, adding or removing the other endpoint can displace or
restore at most one additional nominee. Thus at most three active edges
change: $\{u,v\}$ and at most one further edge incident to each endpoint.
Writing $\mathscr E(D):=\{\{i,k\}:i<k,\ D_{ik}=1\}$ for the edge set
of a graph $D$, this gives
\begin{equation}
\label{eq:app_sir2_active_changes}
 \bigl|\mathscr E(\widetilde A)
       \mathbin{\triangle}\mathscr E(\widetilde A')\bigr|\leq3.
\end{equation}

Next evaluate the epidemic formulas on any two active graphs $D,D'$ of
maximum degree two that differ only in $\{u,v\}$. Only $H_u,H_v$ can
change. Consequently, the only exposures that can change belong to
\[
 \mathcal B_{u,v}
 :=\{u,v\}\cup\mathcal N_u(D\cup D')
              \cup\mathcal N_v(D\cup D').
\]
One of these two graphs contains the other. Their union therefore has
maximum degree two and contains $\{u,v\}$, implying
$|\mathcal B_{u,v}|\leq4$.
This set is independent of the allocation. Every unit outside it has
unchanged exposure for every $\mathbf t$, and hence an unchanged entire
feasible set $C_i$. Its covariate is fixed, so its complete pair
$(L_i,E_i)$ is unchanged as well.

Pass from $\widetilde A$ to $\widetilde A'$ by first deleting edges not
in the latter and then adding its missing edges. Every intermediate graph
has maximum degree two. Equation~\eqref{eq:app_sir2_active_changes}
requires at most three steps, each affecting at most four complete pairs.
Restoring population superscripts, for all $A,A'\in\sA_{n_e}$ differing
in one dyad and all covariate arrays $\mathbf x$,
\begin{equation}
\label{eq:app_sir2_full_diameter}
\begin{split}
 \sup_{\mathbf t\in\{0,1\}^{n_e}}
 \Bigl|\Bigl\{i\in[n_e]:\;&
  \bigl(L_i^{(e)}(A,\mathbf x),
        E_i^{(e)}(\mathbf t;A,\mathbf x)\bigr)\\
 &\neq
  \bigl(L_i^{(e)}(A',\mathbf x),
        E_i^{(e)}(\mathbf t;A',\mathbf x)\bigr)
 \Bigr\}\Bigr|
 \leq\min\{n_e,12\}.
\end{split}
\end{equation}
Here graph--covariate arguments mean that all maintained transformations,
including the active-contact rule, are reevaluated. In particular,
Assumption~\ref{asmp:ub-locality}.b holds in the pilot with
\begin{equation}
\label{eq:app_sir2_constants}
 d_\n^{\mathrm{sc}}=\min\{\n,12\},\qquad
 K_{\mathrm{set},\n}\leq2^{12}-1.
\end{equation}
The catalogue bound follows because every feasible set is a nonempty
subset of the twelve-element exposure space. 
\item \emph{Independent fundamentals and graphon formation.}
For every fixed label $c$, Equation~\eqref{eq:app_sir2_response} has the
form
\[
 Y_i^{(e)}(c)=f(X_i^{(e)},\varepsilon_i^{(e)};c),
\]
where $f$ is common across units and populations. It contains no other
unit's primitive and no dyadic uniform. Therefore
\[
 \tilde F_i^{(e)}
 =\left(U_i^{(e)},X_i^{(e)},
   \{f(X_i^{(e)},\varepsilon_i^{(e)};c):c\in\cC_{\mathrm{SIR}}\}\right)
\]
is a common measurable transformation of the IID primitive vectors
$(U_i^{(e)},X_i^{(e)},\varepsilon_i^{(e)})$.
The complete fundamentals are thus IID, as required by
Assumption~\ref{asmp:ub-fundamentals}.a.
The dyadic uniforms are mutually independent and independent of these
fundamentals by construction, and
Equation~\eqref{eq:app_sir2_graphon} is exactly the required formation
rule. This verifies Assumption~\ref{asmp:ub-fundamentals}.b.
\end{enumerate}

\section{Notation}
\label{app:notation}

The table consolidates related symbols into one entry. Population superscripts are suppressed only when the population is clear. 

{\small\linespread{1}\selectfont
\setlength{\tabcolsep}{4pt}
\renewcommand{\arraystretch}{1.04}
\begin{longtable}{@{}>{\raggedright\arraybackslash}p{0.30\textwidth}
                    >{\raggedright\arraybackslash}p{0.67\textwidth}@{}}
\caption{Principal notation.}\label{tab:notation}\\
\toprule
\textbf{Symbol} & \textbf{Meaning}\\
\midrule
\endfirsthead
\toprule
\textbf{Symbol} & \textbf{Meaning}\\
\midrule
\endhead
\midrule
\multicolumn{2}{r}{\emph{Continued on next page}}\\
\endfoot
\bottomrule
\endlastfoot
$e,\ n_e$ & Population index ($0$ target, $1$ pilot) and size.\\
$X_j^{(e)},\ \bX{e}$ & Unit covariates and their population array.\\
$\bA{e},\ G_j^{(e)},\ \bG{e}$ & Adjacency matrix, observed network feature, and feature array.\\
$V^{(e)},\ \mathcal V_e;\ V$ & Observed state $(\bX{e},\bG{e})$, its space, and $V=V^{(0)}$.\\
$\mathbf t^{(e)},\ \bT{1}$ & Allocation vector and randomized pilot allocation.\\
$B,\ \Tc,\ \Q$ & Budget, feasible target allocations $\{\mathbf t:\sum_jt_j\leq B\}$, and $\Q=|\Tc|$.\\
$\cC^{(e)},\ c_j^{(e)}$ & Finite exposure space and known exposure mapping.\\
$E_j^{(e)}(\mathbf t;v)$ & Exposure at allocation $\mathbf t$ and state $v$; omitting $v$ uses $V^{(e)}$.\\
$C_j^{(e)},\ L_j^{(e)};\ \ell,\mathcal L$ & Feasible exposure set, local state $(X_j^{(e)},C_j^{(e)})$, generic state, and its space.\\
$Y_j^{(e)}(c),\ \Yc{e}_j;\ M$ & Potential outcome, exposure-indexed outcome family, and absolute outcome bound.\\
$\mu_c^{(e)}(\ell)$ & Conditional mean outcome at local state $\ell$ and exposure $c$.\\
$\tilde F_i^{(e)},\ F_i^{(1)}$ & Fundamentals $(U_i^{(e)},X_i^{(e)},\Yc{e}_i)$ and their augmentation by $T_i^{(1)}$.\\
$g,\ \rho_{n_e},\ \xi_{ik}^{(e)}$ & Graphon, link-probability scale, and dyadic uniforms.\\
$E_i^{(1)},\ Z_i,\ \mathbf Z$ & Pilot exposure, score input $(Y_i^{(1)},L_i^{(1)},E_i^{(1)})$, and full pilot data.\\
$p_\n,\ e_{i,c}(\ell)$ & Pilot treatment probability and exposure propensity.\\
$P_L^{(e)},\ P_V;\ s,\kappa$ & Local- and target-state laws; density ratio $s=dP_L^{(0)}/dP_L^{(1)}\leq\kappa$.\\
$\overline\Pi,\ \Pi;\ \pi,\widehat\pi,\widetilde\pi$ & All feasible global rules, learning class, global rule, empirical maximizer, and statistical rule.\\
$E_j^\pi;\ q_j,\overline q$ & Policy-induced exposure; its probability conditional on $L_j^{(0)}$ and the average over $j$.\\
$\mathcal E_\pi(\p),\ \gamma_\J$ & Policy-relevant local-state--exposure pairs and overlap bound.\\
$W,\ \Upsilon_i,\ \widehat W_\n;\ \Pc,R$ & Target welfare, IPW score, empirical welfare; model class and uniform regret.\\
$N_{dim},\ \Pi_0,\ m$ & Natarajan dimension, countable approximating class, and richness dimension.\\
$a_\Pi,\vartheta_\n;\ K_{\mathrm{set},\n}$ & Regret complexity factors; number of distinct pilot feasible-exposure sets.\\
$d_\n^{\mathrm{sc}}$ & Uniform number of local-state--exposure pairs affected by one dyad change.\\
$\lambda,\ \Lambda_B(\p)(\p),\ \pi_\lambda$ & Common local rule, its feasible class, and induced global allocation.\\
$\Delta_\p,\ \mathscr I_j;\ D_j,\mathcal B_j$ & Global-minus-local oracle gap, local information; inclusion advantage and its information-loss bound.\\
$R_j^{\mathrm{soc}},\lambda_{j,B},\Gamma_{j,B}$ & Social return, displaced return, and budget-adjusted return $\Gamma_{j,B}=\J D_j$ under additive LSE.\\
$h,\ h_k;\ W_h,\widehat W_{\n,h}$ & Retention function, degree cutoff, and retained welfare criteria.\\
$\gamma_{\J,h},\beta_h,\overline\beta_h$ & Retained overlap bound, omitted target mass, and its supremum.\\
$b_{\n,k},\ r_{\n,h};\ p_0,\bar\lambda$ & Degree-tail and stochastic bounds; lower assignment probability and upper expected degree.\\
\end{longtable}}

\end{appendix}

\end{document}